\documentclass[cdm]{ipbook}

\RequirePackage{amsthm,amsmath,amssymb}
\RequirePackage{hyperref}
\usepackage{cancel, color}
\hypersetup{
    linktocpage,
    linkcolor={black} 
    }
\RequirePackage{tabularx}

\startlocaldefs
\theoremstyle{plain}

\def\pder#1{\frac{\partial }{\partial #1} }
\endlocaldefs

\usepackage{empheq}               
\usepackage[strict]{changepage}

\usepackage{bm}
\usepackage{enumitem}

\def\Pf{{\bf Pf}}
\DeclareMathOperator{\sgn}{sgn}

\usepackage{mathtools}

\def\Z{\mathbb{Z}}                          
\def\R{\mathbb{R}}  
\def\G{\mathbb{G}}
\def\V{\mathcal{V}}
\def\E{\mathcal{E}}

\def\eps{\varepsilon}

\newtheorem{theorem}{Theorem}[section]
\newtheorem{conj}{Conjecture}
\newtheorem{corollary}[theorem]{Corollary}
\newtheorem{lemma}[theorem]{Lemma}
\newtheorem{proposition}[theorem]{Proposition}
\newtheorem{definition}[theorem]{Definition}

\newtheorem{conjecture}[conj]{Conjecture}

\newcommand{\be}[1]{\begin{equation}\label{#1}}
\newcommand{\ee}{\end{equation}}
\numberwithin{equation}{section}

\newcommand{\ba}[1]{\begin{align}\label{#1}}
\newcommand{\ea}{\end{align}}
\numberwithin{equation}{section}

\newcommand{\ben}{\begin{equation*}}
\newcommand{\een}{\end{equation*}}
\numberwithin{equation}{section}

\newcommand{\calB}{\mathcal{B}}

\newcommand{\calF}{\mathcal{F}}

\newcommand{\bbH}{\mathbb{H}}

\newcommand{\bbZ}{\mathbb{Z}}
\def\Z{\mathbb{Z}}                          
\def\be{\begin{equation}}
\def\ee{\end{equation}}

\newcommand{\n}{\mathbf{n}}
\newcommand{\m}{\mathbf m}

\newcommand{\rk}[1]{\bgroup\color{red}%
  \par\medskip\hrule\smallskip%
  \noindent\textbf{#1}%
  \par\smallskip\hrule\medskip\egroup}

\date{December 5, 2021}

\title[On the  scaling limits of critical Ising and $\varphi^4$ models]
{A geometric perspective on the  scaling limits \\ of critical Ising and $\varphi^4_d$ models}

\author{Michael Aizenman}
\address{Departments of Physics and Mathematics, Princeton University, Princeton NJ, USA}
\email{aizenman@princeton.com}

\begin{document}

\begin{abstract}  
The lecture delivered at the \emph{Current Developments in Mathematics} conference (Harvard-MIT, Jan.  2021) focused on the recent proof of the Gaussian structure of the scaling limits of  the critical Ising and $ \varphi^4$ fields in the marginal case of four dimensions (joint work with Hugo Duminil-Copin).  These notes expand on the background of the question addressed by this result, approaching it from  two partly overlapping perspectives: one concerning critical phenomena in statistical mechanics and the other functional integrals over Euclidean spaces which could serve as a springboard to quantum field theory.    
We start by recalling some basic results concerning the models' critical behavior in different dimensions. 
The analysis is framed in the models' stochastic geometric random current representation.  It yields  intuitive explanations as well as tools for proving a range of dimension dependent results, including: the emergence in $2D$ of Fermionic degrees of freedom, the non-gaussianity of the scaling limits in  two dimensions, and conversely the  emergence of Gaussian behavior in four and higher dimensions.   
To cover the marginal case of $4D$  the tree diagram bound which has sufficed for higher dimensions  needed to be supplemented by a singular correction.  Its presence was established through multi-scale analysis  in the recent  work with HDC.

\end{abstract}

\maketitle

\tableofcontents

\section{Introduction}

Scaling limits of critical systems have been of obvious interest for the theory of critical phenomena in statistical mechanics, and beyond that also for the constructive program of quantum field theory.   
A recent result in this area is  the long awaited proof of the ``marginal triviality'' of the scaling limits of critical $4D$ Ising and $\varphi^4_4$ models.  It was  derived jointly with Hugo Duminil-Copin, and its details are presented in ~\cite{AizDum21}.   As the manuscript  is directly available,  the goal of these notes is to  present some of the background, and place the discussion of four dimensions in the broader context of the  system's different behavior in other dimensions.  

We start with the presentation of the field theoretic perspective on the question addressed in \cite{AizDum21}, and its relation to the natural statistical mechanics question concerning the critical model's fluctuating magnetization field.

 The discussion then pivots to a class of ferromagnetic one component spin systems, which includes both the Ising model and the lattice  $\varphi^4$ fields.  Their critical behavior, which is dimension dependent, can be traced to the  stochastic geometry of  corresponding systems of \emph{random currents} whose definition is recorded below.   Their intersection and percolation properties are in direct relation with the structure of the correlation functions.

To demonstrate the dependence on dimension, and explain the transition which occurs at $4D$  we list here some of  the model's dimension dependent features.    
Specifically (postponing the references, that are provided below):  
\begin{itemize}
  \setlength\itemsep{0.5em}
\item [$\partial$(2D):] 
Along the boundaries of planar Ising models,  the multi-spin correlations are given by Pfaffians of the boundary-to-boundary two point function.  This is indicative of the model's integrability through \emph{emergent} fermionic degrees of freedom.    
\item [2D:]   In general, planar Ising models are integrable, in the sense that the model's free energy can be expressed as the a sum of the contribution of non-interacting degrees of freedom.   Unlike the boundary case, in the bulk these are related to  order-disorder fermionic variables which propagate with complex phases  (expressed in Kac-Ward path amplitudes, Kaufman spinors, Kadanoff order-disorder variables,  Smirnov parafermions).   However, the the model's original spin variables scale into an interactive field theory. 
\item [3D:]   In three dimension the correlations are expected to have as their scaling limit an interactive $\varphi^4$ field theory.  The challenge of elucidating the non-trivial critical behavior in $3D$, which is of obvious interest for statistical mechanics, continues to attract the attention of both  mathematicians and theoretical physicists.
\item [D$\geq$4:] In dimensions $d\geq 4$ the critical models' scaling limits are Gaussian, and as such they correspond to  non-interactive field theories.    
\end{itemize}
While the author's CMD 2020(/21) lecture was focused on the last result, which was derived in the recent joint work with Hugo Duminil-Copin~\cite{AizDum21}, 
the purpose of these notes is to expand on its background and natural context.  \\

The presence of stochastic geometric scaffoldings behind the spin correlations is not a new observation, but even so its range of applicability continues to spread.  Such representations yield insights into duality relations and tools for rigorous results on the critical behavior in models which are out of reach of explicit solutions.   
Given the focus of these notes we shall not delve here into the full range of such methods, and focus on the \emph{random current representation} (RCR).  It provides a useful tool for the analysis of Ising models, and other systems with similar spin-flip symmetry.  
It  allows to  cast in intuitive stochastic geometric terms the results of some of the delicate cancellations which affect the structure of the model's correlation functions. 

The study of correlation functions of critical models on $\Z^d$ and the  statistical field theories which emerge in their scaling limits, has also attracted attention as a possible springboard towards quantum field theory.   Let us start by presenting this relation, borrowing heavily from the introduction of~\cite{AizDum21}.

\section{Constructive quantum field theory  \\  and functional integration}\label{sec:1.2}

\subsection{An outline} \mbox{  } \\[-1ex]  

Quantum field theories with local interaction  play an important role in the physics discourse, where they appear in  subfields ranging from high energy to condensed matter physics.  
The mathematical challenge of proper formulation of this concept led to  constructive quantum field theory programs (CQFT).  A path  to the goal was charted through the  proposal to define quantum fields as operator valued distributions whose essential properties are  formulated as the Wightman axioms~\cite{Wig56}.  Wightman's reconstruction theorem  allows one to recover this structure from the collection of the corresponding correlation functions, defined over 
the Minkowski space-time.  An essential feature of quantum physics is the unitarity of the time evolution operator, generated by a self adjoint operator which, preferably, is bounded below.   A potential link with statistical mechanics is suggested by the  formal step of analytic continuation of the time variable, under which  unitary evolution in time, $e^{it H}$ acting in the suitable  Hilbert space of states,  is related to the semigroup $e^{-\beta H}$ which is of independent interest as the (unnormalized) density operator of the thermal state at  temperature $1/\beta$.   The relation was placed on firm footing by the Osterwalder-Schrader theorem~\cite{OS73,OS75} which lays down a program in which correlation functions 
satisfying the Wightman axioms may potentially be obtained through analytic continuation from those of random distributions defined over the corresponding Euclidean space that meet a number of conditions: suitable analyticity, permutation symmetry, Euclidean covariance, and reflection-positivity.   

More explicitly, seeking natural candidates for scalar \emph{Euclidean fields}, one ends up with the task of constructing  probability averages over random distributions $\varphi(x)$, for which the expectation value of  functionals   $F(\varphi)$ would  have properties fitting the  formal expression  
\be \label{Phi_EV}
 \langle F \rangle  = \frac1{\rm norm}   \int  F(\varphi) \exp[-\mathcal A(\varphi)]   \, d \varphi    ,
\ee 
where  $d \varphi $ is an integral over  (generalized) functions and 
 $ \mathcal A(\varphi) $ is a local `action'.  In this context, one is tempted to consider expressions of the form
  \be \label{eq:FI} 
\mathcal A(\varphi) =  (\varphi, K  \varphi) + \int_{\R^d}   P(\varphi(x)  )  \, dx
\ee
with  $(\varphi, K  \varphi)$ a positive definite and reflection-positive quadratic form, and 
$P(\varphi  )$ a polynomial (or a more general function) whose terms of  order $\varphi^{2k}$ are interpreted heuristically as representing $k$-particle interactions.   An example of a quadratic form with the above properties  (at $J, b>0$) and also rotation invariance is  
\be \label{free_field}
  (\varphi, K  \varphi)  =  \int_{\R^d}  \left( J \| \nabla \varphi(x) \|^2 + b \vert \varphi(x) \vert^2 \right)   \, dx.
\ee

The functionals $F(\varphi)$ to which \eqref{Phi_EV} is intended to apply include  the smeared averages  
\be\label{eq:functional}
\Phi_f(\varphi) :=  
\int_{\R^d} f(x) \varphi(x) dx
\ee
 associated with continuous functions of compact support $f\in C_0(\R^d)$.    
By linearity, the expectation values of products of such variables take the form
\be \label{eq:Schwinger} 
\boxed{ \langle \prod_{j=1}^{n} \Phi_{f_j} \rangle   :=   \int_{(\R^d)^n}  d x_1 \dots d x_n  \, S_n(x_1,\dots, x_n) \,\prod_{j=1}^n  f(x_j) }
\ee 
with  $S_n(x_1,\dots, x_n)$  characterizing the  probability measure on the space of distribution which corresponds to the expectation value $\langle \,\cdot\,   \rangle$.    This is summarized by saying that 
 in a distributional sense 
\be
\langle \prod_{j=1}^{n} \varphi(x_j) \rangle = S_n(x_1,\dots, x_n)   \, , 
\ee
with $S_n$ referred to as the {\em Schwinger functions} of the corresponding euclidean field theory.  

The joint distribution of the fields can also be described in terms of the Ursell functions $U_n$  
which bear the following relation to  the moment generating functional 
\be  \label{U_def} 
\langle \exp \left \{ \Phi_f \right \} \rangle = \exp\left\{\sum_{n=1}^\infty \frac{1}{n!} \int_{(\R^d)^n}  d x_1 \dots d x_n  \, U_n(x_1,\dots, x_n) \,\prod_{j=1}^n  f(x_j), 
 \right\} 
\ee 
for all test functions $f\in C_0(\R^d)$. 

A relatively simple class of Euclidean fields are the Gaussian fields, which result from $H$ with  only quadratic terms.   More generally, Gaussian fields of flip-symmetric distribution (with $S_{2n+1}(x_1,...,x_{2n+1}) \equiv 0$) are characterized by: 
\be 
 U_{2n}(x_1,..., x_{2n}) = 0 \quad \mbox{$\forall n\geq 2$\,, \quad and \quad} U_2(x_1,x_2) = S_2(x_1,x_2)  \,
\ee 
or, equivalently, the  Wick's law:
\be \label{triv}
\boxed{S_{2n}(x_1,\dots,x_{2n})  =  \sum_{\pi}  \prod_{j=1}^n S_2(x_{\pi(2j-1)},x_{\pi(2j)})\, =: \, \mathcal {G}_n[S_2](x_1,\dots,x_{2n}) 
}  
\ee 
where $\pi$ ranges over pairing permutations of $\{1,\dots,2n\}$.   

The presence  in the continuum theory of  particle interactions is expressed in the deviation from Wick's law.  In situation in which  $S_2(x,y)$ has the interpretation of a propagation amplitude,  $U_{2n}$ can be viewed as representing  effective $n$-particle interactions.   Gaussian fields, for which such multi-particle interactions are absent, have on occasions been referred to as {\em trivial}.\footnote{{\em Triviality} is a misnomer when applied to a scaling limit, since even in such cases the determination of the limit's  two point function $S_2$ would require some non-trivial analysis.} \\

However, even in the simple looking case of the Gaussian \emph{free field},  with   $H$ consisting of just the quadratic term \eqref{free_field},  one finds that \eqref{Phi_EV} requires some care.   It cannot be interpreted literally 
since ipso-facto the resulting measure is supported on non-differentiable functions for which the integral in the exponential, evaluated through the function's Fourier transform, is almost surely divergent.  Furthermore, in dimensions $d>1$, the function $\varphi$  itself  fails to be locally square integrable.   Generally  known methods allow to tackle  these difficulties.  However, beyond one dimension (in which case one gets the Wiener measure)  the result is not  a random function but a random distribution.  This in turn complicates the interpretation of the next polynomial term in \eqref{eq:FI}.

The heuristic ``renormalization group'' analysis by Wilson~\cite{Wil71} indicates that in low enough dimensions, specifically $d<4$ for $\varphi^4$,  the difficulties can be tackled through cutoff-dependent counter-terms.  
In particular, to produce what would be regarded as a $\varphi^4$ field theory the local interaction term in  \eqref{eq:FI}  is taken to be of the form 
\be \label{eq:P4}
 P(\varphi) =\lambda \varphi^4\,-b\varphi^2 \,, 
\ee
allowing for $\lambda$ and $b$ to be adjusted as the \emph{ultraviolet} (finite volume) and \emph{infrared} (short distance) cutoffs  cutoffs are removed (the first $R\to \infty$ and the second $a\to 0$).   
The goal of this process is to produce convergence of the Schwinger functions $S_n(x_1,.., x_n)$ to a non-trivial limit.   

\subsection{Lattice regularization} \mbox{  } \\[-1ex]  

Partially successful attempts to carry rigorously a constructive project along the above lines  have been the focus of a substantial body of works.    The means employed have included various versions of rigorous renormalization group analysis, with counter-terms, which are allowed to depend on regularizing cutoffs, scale decomposition, renormalization group flows, and/or regularity structures.   

In our discussion we focus on the attempts to construct a $P(\varphi)$  expectation value functional  \eqref{Phi_EV} 
as the  limit of its discretized  versions
\be \label{Phi_SM}
\boxed{\langle F(\varphi) \rangle  = \frac 1 {\rm norm}   \int  F(\varphi) \exp{[-H(\varphi)]}  \prod_{x\in V_{a,R}}
e^{- \lambda \varphi_x^4 + b \varphi_x^2}  d \varphi_x 
 .}  
\ee 
Here $\{\varphi_x\}_{x\in V_{a,R}}$ is a collection of variables associated with the sites of the finite graph
\be \label{eq:a_R} 
\mathcal V_{a,R} = (a \Z)^d \cap [-\mathcal L,\mathcal L]^d \,  
\ee 
with an  a-priori distribution of the form $\rho(d \varphi_x) = e^{- \lambda \varphi_x^4 + b \varphi_x^2}  d \varphi_x $ 
and $H$ is  the short range interaction
\be \label{H1}
H(\varphi) = -\frac{1}{2} \sum_{\{x,y\}\subset \mathcal V_{a,\mathcal L}}J\,(\varphi_x-\varphi_y )^2    \, . 
\ee 

Quantities whose joint distribution we track  in the scaling limit are based on the collections of random variables of the form
\be \label{def_Tf_scaled_FT}
\boxed{T_f(\varphi) :=  \frac{1}{\sqrt{\langle \big[\sum_{x\in a\Z^d \cap [-1,1]^d}  \varphi_x \big]^2\rangle}}\sum_{x\in\mathcal V_{a,\mathcal L} } f(x ) \, \varphi_x  }  
\ee
where $f$ ranges over $C_0(\R^d)$ - the collection of compactly supported continuous functions.

\begin{definition} Given a sequence of joint values of $(a, \mathcal L,  J,\lambda, b)$ for which $a \to 0$ and $ \mathcal L \to \infty$, we say that model  scaling limit exists  if for any finite collection of test functions $f\in C_0(\R^d)$ the joint distributions of the random variables $\{T_{f}(\varphi)\}$ has a proper limit (a probability distribution with respect to which the variables are almost surely finite).  
\end{definition}

Through a standard probabilistic construction,  the limit can be presented as a random field $\varphi$.
\footnote{By the Kolmogorov extension theorem, one may start by selecting  sequences of the parameter values so as to establish convergence in distribution for a countable collection of test functions $f$, which is dense in $C_0(\R^d)$, and then use the uniform local integrability of the rescaled correlation function and of the limiting Schwinger functions, to extend the statement by continuity arguments to all $f\in C_0(\R^d)$.  One may then  recast the limiting variables as associated with a single random $\varphi$, as in \eqref{eq:functional}.}\\ 

A  regularity condition which it is natural to add here is that  $S_2(x,y) < \infty $ for non-coincidental pairs ($x\neq y$) and 
\be 
\lim_{\|x-y\|\to \infty } S_2(x,y) = 0
\ee

A point of fundamental importance here is that the correlations of the resulting field are observed on a scale ($L$) with is much larger than $a$ -- the range of the interactions due to which they emerge.    For that to occur, the 
system's parameters $(J, \lambda, b)$ need to be very close to the critical  manifold along which the lattice system's correlation length $\xi$ diverges (on the scale set by $a$.)   

Thus, the scaling limits discussed here are constructed the continuum limits of the  distribution of the variables $T_{f,L}$ induced by \eqref{Phi_EV} in different limits, for which 
the order of the different scales is 
\be \label{scales} 
a  \ll L \ll \xi(\beta) \leq \mathcal L \, .     
\ee 
where, to summarize: $a$ is the lattice scale,  $L$ the continuum limit scale, i.e. the scale which is prodded through the variables $T_{f,L}(\sigma)$,  $\xi$ is the correlation length,  $\mathcal L$ is the finite systems cutoff,   and 
$A\ll B$ means $A/B \to 0$.   \\

The constructive field theory   
program  has yielded non-trivial scalar field theories $\varphi^4_2$ and $\varphi^4_2$ (where the subscript is the space dimension and the upper script the degree of the polynomial $P(\varphi)$)~\cite{BryFroSpe82,GliJaf73,GRS75,OS75}.   
The  progression of constructive results was halted when it was proved  that for $d>4$ 
the scaling limits of the field's lattice versions at $\beta \leq \beta_c$   
yield only  Gaussian fields~\cite{Aiz82,Fro82}.  (Gauge field theories are not discussed here, cf.~\cite{JW}).
\\

In the $80$'s various partial results have indicated that the same may hold true for the critical dimension $d=4$  
(cf.~\cite{AizGra83, ACF83, BauBrySla14,GawKup85, HarTas87}), however a  proof which covers also the strong coupled $\varphi^4$ (of which Ising is an example) was reached only in the recent work \cite{AizDum21}.   
Its  main result is the completion of the following statement  which now covers also the marginal (and thus the hardest) case $d=4$.

 \begin{theorem}[Gaussianity of $\varphi^4_d$, $d\geq 4$]\label{thm:gaussian phi4}
For  any dimension $d\geq 4$,  any random field reachable by the above construction as a scaling limit of measures of the form \eqref{Phi_EV} in limits meeting the condition \eqref{scales}  is a generalized Gaussian process. \end{theorem}

The previous results concerning $4D$   include proofs of convergence to Gaussian field if one starts from small enough $\lambda$, i.e. sufficiently small perturbations of the Gaussian free field\cite{BauBrySla14,FMRS87,GawKup85,HarTas87}.  The   
rigorous renormalization techniques presented in these works  yield also more detailed asymptotic, for that case.    In comparison, the result proven in \cite{AizDum21}  covers  arbitrarily ``hard'' $\varphi^4$ fields (i.e. large $\lambda$).  However, it does not currently provide comparable information on of the exact scale of the logarithmic corrections, and the exact expression for the covariance of the limiting Gaussian field.

To keep the presentation simple, we first discuss the standard Ising model, and then present the way the analysis extends to ferromagnetic system of lattice $\varphi^4$ variables. \\

\section{The statistical mechanical perspective} 

What from the perspective of constructive field theory may be taken as a disappointment is a positive and constructive result  from the perspective of statistical mechanics.   The theoreticians' goal there is to understand the critical behavior in models which lie beyond the reach of exact solutions.   The proven gaussianity of the critical model's scaling limits in $d\geq 4$ dimensions has therefore also a positive aspect.

Among the simply stated examples of the mean value functional \eqref{Phi_SM} are the scaled versions of the Gibbs equilibrium measures of the Ising model on $\Z^d$.

A configuration of Ising spin variables on a finite graph $\G=(\mathcal V, \E)$ (written in terms of its vertex set and  edge set)  is a collection of binary variables $\{\sigma_x\}_{x\in\mathcal V} $ which define a  function $\sigma : \V \mapsto \{-1,1\}$.   The standard presentation of its Hamiltonian is 
\be \label{H_I}
\boxed{{H}_\G(\sigma) := - \sum_{\{x,y\} \in \E} J_{x,y} \sigma_x \sigma_y - h \sum_{x \in \V} \sigma_x }  
\ee
with $J:=\{J_{x,y}\}_{\{x,y\}\in \mathcal E}$ the {\em coupling} coefficients, and $h$ is the {\em magnetic field parameter}.   

The system's Gibbs equilibrium states at inverse temperature $\beta  $ correspond to expectation value averages with respect to  probability measures on $\{-1,1\}^{\V}$ for which 
\be 
\boxed{ \mu_{\V,\beta}(\{\sigma\}) = e^{-\beta H(\sigma)} / Z_{\mathcal V}(\beta) }\,.    
\ee
The normalizing factor is  the {\em partition function} 
\be \label{Z}
 Z_{\mathcal V}(\beta) \  := \  \sum_{\sigma \in \{-1,+1\}^{\V}} \exp(-\beta {H}_G(\sigma))
\ee 
(An attentive reader may note the small discrepancy between \eqref{H_I}  and \eqref{H1}.  The  difference amounts to purely quadratic terms which even when not constant can be absorbed in the definition of the a-priori measure.) \\ 
The Gibbs measure for the infinite graph $\Z^d$ is defined as the $\mathcal L \to \infty$ limit of the above in 
$\mathcal V = [-\mathcal L,\mathcal L]^d$.   

For generic values of the temperature, and/or the interaction parameters   $J_{x,y}$ the spin variables' covariance decays exponentially in the distance, with a finite correlation length $\xi(\beta)$.    However, at $h=0$  the model undergoes a phase transition which is {\it continuous} in the sense that   $\lim_{\beta \searrow \beta_c} \xi(\beta)=\infty $.   
At the critical point the system's spin fluctuations are meaningfully correlated at all scales.     
One way to described the correlation's structure is through the field theory which emerges as a scaling limit of the variables 
\be \label{def_Tf_scaled}
\boxed{T_{f,L}(\sigma) :=  \frac{1}{\sqrt{\langle \big[\sum_{x\in \Z^d \cap [-L,L]^d}  \sigma_x \big]^2\rangle}}\sum_{x\in \Z^d } f(x/L ) \, \sigma_x   }
\ee
where $f$ ranges over $C_0(\R^d)$ - the collection of compactly supported continuous functions.  Looking up from the perspective of the lattice the relation of the different scales is as presented in \eqref{scales}, with $a=1$. 

For the Ising model at $h=0$, in which case the only adjustable thermodynamic  parameter is $\beta$, it is expected that the scaling limits  will be the same for all choices of $L, \mathcal L \to \infty$ and $\beta \searrow \beta_c$ meeting the condition \eqref{scales}. 

Ising and $\varphi^4$ lattice systems can each be presented as the limit of the other.  The direction ``$\varphi \Rightarrow$  Ising'' is rather obvious.  The somewhat more surprising converse relation is through the Griffiths-Simon representation of real variables with the  $\varphi^4$  distribution as a limit (in distribution) of the empirical average of ferromagnetically coupled elemental Ising spins~\cite{Gri69,SimGri73}.   The collection of such variables forms what we call the Griffiths-Simon class.  

The more explicit definition of the Griffiths-Simon class of variables is postponed here to section~\ref{sec_GS}, even though references will be made to it earlier in the text.  The most relevant point about it are summarized above.   \\ 
 
The analysis of Ising models is facilitated by the random current representation  (RCR).  
Its power derives from a combinatorial identity (the ``switching lemma'') which made its appearance in the Griffiths-Hurst-Sherman work on  their eponymous correlation inequality~\cite{GHS70}. 
It was further  developed in \cite{Aiz82} as a tool for expressing the  spin covariances in stochastic geometric terms.
   
Combined with simple geometric considerations, RCR allows to present  the Ursell function $U_4$ of \eqref{U_def} as  a multi-site connection amplitude for random currents linking prescribed sources.   
From this follow (omitting here the credits, which are given in the statements' more detailed presentations in the corresponding sections, below): 

\begin{itemize} 
  \setlength\itemsep{0.5em}
\item A simple explanation of the  Pfaffian  and fermionic relations which appear in the planar case, and a positive elementary lower bound on $|U_4|/S_4$ in $d=2$ dimensions.
\item An intuitive explanation, alas not yet a proof, of the non gaussian nature  of the scaling limits in $d=3$ dimension.
\item A relation between truncated correlations which quantifies the statement that the model's scaling limit is Gaussian if and only if Wick's law holds that the level of the $4$-point function.  
\item  A diagrammatic upper bound on $U_4$ in terms of the model's two point function.  When combined with  the ``infrared bound''  on the latter, it yields a simple proof of the triviality of the scaling limits in $d>4$ dimensions.
\item  A stepping stone for the extension of the above to  $d=4$.
\end{itemize} 
For the latter, an essential use was made of the fact that the random current representation produces not only a diagrammatic upper bound on $U_4$ but actually an identity from which one may extract also a lower bound.  Both play a role in the multi scale analysis which was needed for the marginal dimension.   

The random current representation was also instrumental in establishing what may be viewed as a more preliminary statement about the model's phase structure, which are recalled below. 
  
  \section{Preliminaries: the ferromagnetic models' phase transition} \label{sec:preliminaries}

Statistical Mechanics, has had as its initial success the explanation of the roots of the observed thermodynamic behavior: the meaning of entropy (Boltzmann), the roots of irreversibility and the phenomena of phase transition.   Its impact has however extended quite further: it played an essential role in the emergence of quantum physics (Plank), information theory (Shannon), financial market's analysis (spin glass), virology (contact process), complexity in classical and quantum computations, and current attempts to conceptualize the fabric of space-time (SYK).  

Ising spin systems, through various generalizations of the initial model, play a similarly fecund role within statistical mechanics.  They offer an arena for the formulation and testing of a range of concepts and analytical tools, only few of which will be mentioned here.

The microscopic structure of real matter is obviously more complex than that of the Ising model.  And so is the phase structure of water and of materials exhibiting  ferromagnetic transition.  Yet it is an observed fact that at continuous phase transitions the critical behavior in  many different systems agree to an astounding degree (upon smooth adjustments of the thermodynamic parameters).   
A common explanation is that since the scale on which the systems are observed  is vastly larger than  that of its microscopic constituents, the scale invariant critical correlation functions may be governed by continuum  theories which emerge from local interactions, and these come with a  limited collection of relevant parameters.   

This observation gives impetus to the study the phase transitions in models of simplified microscopic structure, such as the model which Ising was invited to study by his research advisor Lenz. 
Its basic setup is described in \eqref{H_I}.  

We shall focus here on the models formulated on the $d$-dimensional graphs $\Z^d$, with the translation invariant nearest neighbor interaction.   However it may be of interest to note that 
some of the basic results on the Ising models' phase structure are valid in the greater generality of transitive  (i.e. invariant under a transitive group of automorphism) and ameanable graphs,  
with homogeneous interaction (by which is meant here invariance under a transitive group of  graph homomorphisms).  
 \\ 
 
 \bigskip    

\subsection{Sharpness of the phase transition and general critical exponent bounds} 
\mbox{} \\

A natural starting point is the following  pair of general statements on the phase structure of spin models in the class described above.

\begin{theorem} 
For any Ising model on a transitive and amenable graph, with a homogeneous  interaction \eqref{H_I}  
with  $\|J\|_1:= \sum_x J_{x,x_0} <\infty$,  
\begin{enumerate}
\item [1)] (Existence of the limit~\cite{Gri67,FKG})\\
  The limit defining the Gibbs state (formulated here with the free boundary conditions) exists for all $(\beta,h)\in [0,\infty)\times\R$
\item  [2)] (Lee-Yang analyticity~\cite{LeeYang52}) \\ 
For any $\beta <\infty$ the state is real analytic in $h$ over $\R\setminus \{0\}$
\item [3)]  (Analyticity at high temperatures~\cite{DobShl87})   \\ 
 For $\beta < 1/  \left[\| J_{0,x}\|_1 \right] $  the state is a real analytic in $(\beta,h)$, in the sense that so are all the expectation values $\langle F \rangle_{\beta,h}$ of  local measurable functions of the spins. 
\end{enumerate} 
\end{theorem} 

 A relevant implication of (2) is that phase transitions in these models can occur only along the line $h=0$.  
In dimensions $d>1$ at low enough temperatures the state is in fact discontinuous across that line, i.e the system exhibits a \emph{first order} phase transition.  The proof of this fact is a celebrated results of Peierls~\cite{Pei36}, in which he has elevated the model's status and confirmed the hopes of its creator.   

The discontinuity is manifested in a jump in the magnetization $M(\beta,h):= \langle \varphi_0 \rangle_{\beta,h}$, or equivalently in the strict positivity of the \emph{residual} (or \emph{spontaneous}) magnetization  
\be 
M^*(\beta) = \lim_{h\searrow 0}  M(\beta,h)  \,.
\ee   
Griffiths' correlation inequalities imply that  $M^*(\beta) $ is monotone increasing in $\beta$ and thus the first order phase transition occurs along the line segment   $\beta\in (\beta_c, \infty]$.\footnote{As to the situation at $\beta_c$: with one notable exception~\cite{ACCN,Tho69,AndYuvHam70}  it is expected that $M(\beta_c,0)=0$.  Proven cases include the nearest neighbor models on $\Z^d$ in all  dimensions~\cite{AizFer86,AizDumSid15} ($3D$ being the last to be proven).}   

While the model's high and low temperature regimes can be prodded by convergent expansions,  its  behavior close to $(\beta_c,0)$ can be accessed only through non-perturbative means.  
No exact solution was found  beyond the two-dimensional nearest neightbor case.  Nevertheless, the following could be established rigorously (joint works with Barsky and Fernandez~\cite{AizBarFer87}).   

\begin{theorem}[\cite{AizBarFer87}] \label{thm:sharpness} 
For any ferromagnetic Ising model on $\Z^d$ with a shift invariant pair interaction, of exponential decay in $|x-y|$,
there exists $\beta_c \leq \infty$ such that: 
\begin{eqnarray} 
 \forall \beta < \beta_c:  &   \quad  \,\,
\langle \varphi_x \varphi_y \rangle_{\beta,0}   \leq A(\beta) \, e^{-\|x-y\| /\xi(\beta)} 
 & \quad \mbox{with $A(\beta), \xi(\beta) <\infty$}    \notag \\   
\label{eq:offcriticality} \\ 
  \forall \beta > \beta_c:  &    \! \! \! \! \! \! 
 \langle \varphi_x \varphi_y \rangle_{\beta,0}  \geq M^*(\beta)^2   
& \quad \mbox{with $M^*(\beta) >0$}    \notag  \,.   
\end{eqnarray} 
And,  approaching the critical point $(\beta_c,0)$ from three different directions:
 \be \label{power_laws}
\boxed{ M(\beta,0)   \geq     \rm{C} |\beta -\beta_c|_+^{1/2} \, , \quad   
M(\beta_c,h)  \geq   \rm{C} \, h^{1/3} \,, \quad
\frac{\partial M}{\partial h}(\beta,0)  \geq   \frac{\rm{C}}{|\beta_c-\beta|_+^{1}}  }
\ee
with constants $\rm{C}$ which may vary from one equation to another.
\end{theorem} 
  
Though stated initially for $\Z^d$, the given proofs readily generalize to   Ising spin systems on amenable transitive  graphs with homogenous ferromagnetic interaction of exponential decay.  Furthermore, for the above results - with the exception of  the exponential decay at $\beta <\beta_c$, the assumption  of exponential decay of $J_{x,y}$  can be relaxed to just  $\|J\|_1<\infty$. \\ 

In essence, this general result can be seen as conveying two statements.    The first is that  in the above described Ising models the phase transition occurs directly from the regime of exponential decay of correlations to that of long range order.  That has been referred to as  {\it sharpness  of the phase transition}.  

Its second part implies  that the model's  
{\it critical exponents}\footnote{Allowing for the presence of logarithmic corrections we define the  critical exponents  cautiously  as the limits of the corresponding log ratios, e.g. 
$
\delta^{-1} = \lim_{h\searrow 0} \,  \, \log M(\beta_c,h)/ \log h   $.
} 
 are bounded by what turns to be their values in the 
model's mean-field version  (for which they are easily calculable)  
\be \label{MF_exp}
\hat{\beta}  \leq 1/2\, , \qquad\delta \geq 3\, , \qquad   \gamma\geq 1 \, \,.   \\ 
\ee

The proof of Theorem~\ref{thm:sharpness} proceeds through partial differential inequalities (PDI), listed in section~\ref{sec_PDI}, which were derived for the order parameter $M(\beta,h)$.     These build on relations which are derived through the model's stochastic geometric representation.   Partly overlapping results were developed for the independent percolation models, and the comparison of the two has been quite instructive.   

Let us  add that another path to the sharpness of the phase transition was more recently presented in \cite{DRT19_sharp}.  It applies more generally to $Q$-state FK random cluster models, a class which includes both Ising ($Q=2$) and independent percolation ($Q=1$).  \\

As is also explained below, for the Ising model with the nearest neighbor interaction  the  {\it critical exponent bounds}   \eqref{MF_exp} are  
proven to be {\bf saturated} in dimensions $d\geq 4$~\cite{AizFer86}.   
The existence of an {\bf upper critical dimension},  above which the critical behavior simplifies is a phenomenon shared also by some other statistical mechanical models.  In particular,  for the independent percolation the corresponding  critical exponent bounds are {saturated} in dimensions $d\geq 6$ (through the combinations of the results of \cite{AizNew,BarAiz91,Har_vdHT_Sla03}).  
 \\ 

\subsection{Bounds on the two point function at $\beta_c$}  \mbox{} \\

In the analysis of the critical behavior in high dimensions an important role is played by the following bounds on the decay rate of the spin-spin correlations at the critical point.   
\begin{theorem}\label{thm_S2bounds}
For the  Ising model on $\Z^d$, $d>2$, with the nearest neighbor interaction  ($J_{u,v} = J \delta_{\|u-v\|, 1}$), at the critical point
\be \label{crit_decay}
\boxed{ \frac{\rm{C_1}}{\|x\|^{d-1}}  \leq \beta J \, \langle \sigma_x \sigma_y \rangle_{\beta_c,0}   \leq \frac{\rm{C_2}}{\|x-y\|^{d-2}} } 
\ee
\end{theorem}
The very important upper bound is based on the model's reflection positivity, through either the transfer matrix analysis of Glimm and Jaffe~\cite{GliJaf73} or (with better constant)  the \emph{Gaussian domination}  bound of Fr\"ohlich, Simon and Spencer~\cite{FroSimSpe76}.   For the spacial form of the upper bound in \eqref{crit_decay}  the original's Frourier-transform is  combined with the Messager Miracle-Sole monotonicity~\cite{MesMir77}, as was done by Sokal~\cite{Sok79}.

The lower bound in \eqref{crit_decay} is due to Simon~\cite{Sim80}, through a  generally applicable strategy with roots  in the work of Dobrushin~\cite{DobShl87} and, going further back,  Hammersley~\cite{Ham57} (in the context of percolation models).

Combined with \eqref{eq:offcriticality} these bounds  demonstrate the point made above: power law behavior of $\langle \varphi_0 \varphi_x \rangle_{\beta,0}$, which is essential for scaling limits, can be found in the vicinity the model's critical point and only there. The extension of this statement to the truncated correlations at the high $\beta$ phase as was recently presented in~\cite{DumGosRao18}.   

In the analysis of the marginal case of  $4D$ we added to \eqref{crit_decay} the observation that 
$ \langle \varphi_0 \varphi_x \rangle_{\beta_c,0}   \cdot \|x\|^{d-2}$ is {\bf monotone decreasing} in the scales of $\|x\|$~\cite{AizDum21}.  
(This played a useful role in combining the contributions of different scales to the effective coupling in $4D$, resulting in a logarithmic-type correction to the previous upper bound by $C/L^{d-4}$.)\\

Let us add that for $d>4$ the lace expansion method is getting ever closer to a proof that the upper bound in \eqref{crit_decay} is assymptotically  saturated~(cf. \cite{Sak21} and references therein).   \\ 

\bigskip 

\section{Ising model's random current representation}\label{sec:RCR}

\subsection{Definition and the switching lemma}\mbox{  }\\

In the random current representation (RCR), the Ising model's defining degrees of freedom $\{\sigma_x\}_{x\in \V}$ are summed over and replaced by integer valued variables $\{n(x,y)\}_{(x,y)\in \E}$ associated with the graphs edges.    The transition casts the partition function as that of a system of loops. The spin variables reappear as source insertion operators in a systems or random currents with source constraints.  

Due to  the positivity of the RCR measure, and its convenient combinatorial properties,  this representation  uncovers a stochastic geometric  underpinning of the spin-spin correlation functions, and sheds lights on the structure of the model's correlation functions, and by implication of the related field of spin fluctuations. 

In formulating the RCR we focus on the spin-flip  symmetric Ising Hamiltonian \eqref{H_no_h} with $h=0$, i.e. at zero magnetic field.  (The extension to $h\geq 0$ is not difficult, and is also useful~\cite{AizBarFer87}) \footnote{A natural extension of RCR to $h\geq 0$ is   obtained through the addition of a ``ghost site'' -- a device which was first employed  by Griffiths in his derivation of correlation inequalities for odd numbers of spins~\cite{Gri67}.}

\begin{definition} For an Ising model on a graph $\G=\{\mathcal V,   \mathcal E \}$, with the Hamiltonian 
\be \label{H_no_h}
{H}_\G(\sigma) := - \sum_{\{x,y\} \in \E(\G)} J_{x,y} \sigma_x \sigma_y 
\,. \ee
a {\em current} configuration $\n$ is an integer-valued function $\n: \mathcal E  \to \Z_+$.  
Referring to the edges as $(x,y)$, the current's set of {\em sources} is the set of sites
\be 
\boxed{ \partial\n \, := \, \{ x\in \V :\, (-1)^{\sum_{(x,y) \in \E} \n(x,y)} = -1 \} }
\ee  
At given $\beta \geq 0$,  we associate to the current configurations  the {\em weight} function
\be 
\boxed{ w(\n)  :=\prod_{\{x,y\}\subset \E}\frac{\displaystyle(\beta J_{x,y})^{\n(x,y)}}{\n(x,y)!} } 
\ee
\end{definition}

\begin{lemma} 
For any $\beta \geq 0$ the   Ising model's partition function (at $h=0$) can be expressed as the following sum over {\bf sourceless} random currents
\begin{equation}\label{eq:8}
Z (\Lambda,\beta) :=\sum_{\sigma: \V\to \{-1,1\}} \prod_{(x,y)\in \E} \exp(\beta J_{x,y}\sigma_x\sigma_y) 
=
 \boxed{ 2^{|\Lambda|} \sum_{\substack{\n:\,\E\to \Z_+ \\ \partial\n=\emptyset}} w(\n)} \,.
\end{equation}

Furthermore, the  Gibbs state expectation value of even  products of spins can be presented as the effect on the sum of the insertion of sources at the corresponding sites:   
\begin{equation}\label{eq:erg}
\boxed{ \langle\prod_{x\in A}\sigma_x\rangle_{\Lambda,\beta}=\frac{\displaystyle\sum_{\n:\, \partial\n=A}w(\n)}{\displaystyle\sum_{\n:\partial\n=\emptyset}w(\n)}\,  }
\end{equation}
(the odd correlations vanish by the model's symmetry).  \\ 
\end{lemma}
 
The proof is elementary, starting from Taylor's expansion  
\be \label{RCR_spin_product}
\exp(\beta J_{x,y}\sigma_x\sigma_y)=\sum_{\n(x,y)\ge0} \frac{(\beta J_{x,y}\sigma_x\sigma_y)^{\n(x,y)}}{\n(x,y)!} \, .
\ee

It should be noted that any sourceless configuration, i.e. one with $\partial\n=\emptyset$, can be viewed as the edge count of a multigraph which is decomposable into a union of loops.  Configuration with $\partial\n=A$, such as the one appearing in the numerator of \eqref{eq:erg}, can be viewed as describing the edge count of a multigraph which is decomposable into a collection of loops and of paths connecting pairwise the sources, i.e. the sites of $A$.  

For any $\beta$  the insertion of  sources at $A=\{x,y\}$ produces a locally detectable change in the constraints on $\n$.   In addition, it requires the admissible multigraphs to link the two sites.  Since at small $\beta$ typical graphs are sparse,  the latter condition should be expected to produce an exponentially decaying factor.  In contrast, at high $\beta$, where one may expect robust percolation, the existence of a long connection would not in itself cause a detectable effect.   

More explicit statements can be reached through  combinatorial identities, such as the one presented by Griffiths-Hurst-Sherman~\cite{GHS70} in their derivation of the GHS inequality.   The following formulation is particularly suited for our purpose~\footnote{A simple way to process the combinatorics underlying Lemma~\ref{lem:switching}  is through the multigraph representation of the combination of the two random currents, cf. \cite{GHS70, Aiz82}.  The  switching argument can, and has been, deployed in a number of slightly varying setups,.} 

\begin{lemma}[Switching lemma]\label{lem:switching}  Let  $\G_1 \supset \G_2$ be a pair of graphs, one included in the other, and $J: \mathcal E_1 \to \Z_+$ a set of coupling constants which for the two graphs coincide on shared edges.   Then for any pair of vertex sets $A\subset \V_1$ and $B\subset \V_2$  
\begin{multline} \label{eq:switching}
\sum_{\substack{
\n_1:\partial\n_1=A\\
\n_2:\partial\n_2=B}}F(\n_1+\n_2)w(\n_1)w(\n_2) =  \\ 
\sum_{\substack{
\n_1:\partial\n_1=A\Delta B\\
\n_2:\partial\n_2=\emptyset}}F(\n_1+\n_2)w(\n_1)w(\n_2)\bm 1_{{\n_1+\n_2}\in \calF_B}  .
\end{multline}
where $\n_j$ are summed over the random current configurations on the corresponding graphs $\G_j$, 
$\calF_B$  is the set of currents $\m$ on $\G_1$ for which there exists a sub-current  $\n' \leq \m$ supported on  $\E_2$ with $\partial\n'=B$, and   $A\Delta B$  denotes the symmetric difference of sets ($:= (A\setminus B)\cup(B\setminus A)$). 
\end{lemma}  

The representation \eqref{eq:erg} and  related combinatorial identities were employed in Griffiths-Hurst-Sherman's proof of the concavity  of the magnetization $M(\beta,h)$ as a function of $h$ over $h\geq 0$.
The term 
\emph{random currents} was coined in \cite{Aiz82}  in the process of developing a stochastic geometric perspective on the structure of the spin-spin correlations between distant sites.

\subsection{Spin correlations as connectivity amplitudes} \mbox{ }\\[-1ex] 

The switching lemma allows to establish some exact cancellations though which emerges  a picture of the stochastic  geometry underlying the  structure of the model's  spin-spin correlations.     Following are three such relations, each derived by an elementary application of Lemma~\ref{lem:switching}.  \\  

For this purpose we denote :  \\ 

\noindent \mbox{ } \quad ${\bf P}^{A_1,...,A_k}_{\Lambda,\beta}[\dots] $ --  the probability average over independent $k$-tuples\\  
\mbox{ } \hspace{3cm}  of currents  $\{\n_1,..., \n_k\}$ with  source sets $\partial \n_j = A_j$\\[1ex] 
\noindent \mbox{ } \quad ${\bf E}^{A_1,...,A_k}_{\Lambda,\beta}[\dots] $ --  the corresponding expectation value \\[1ex]  
\mbox{ } \quad $\bm{1}[x\stackrel{\m}{\leftrightarrow} y]$\quad -- the indicator function of the event that  $x$ and  $y$  \\ 
\mbox{ } \hspace{2.5cm} are connected  by a path along  edges  of non-vanishing $\m$  \\[2ex]

\noindent  {\bf i.  The onset of long range order as a percolation transition}
\begin{equation}\label{eq:xy}
\boxed{ 
\left[ \langle\sigma_{x} \sigma_y \rangle_{\Lambda,\beta}\right]^2=
\frac
{\displaystyle\sum_{\substack{\n_1:\partial\n_1=\emptyset \\ \n_2:\partial\n_2=\emptyset}} w(\n_1) w(\n_2) \bm{1}[x\stackrel{\n_1+\n_2}{\longleftrightarrow} y] }
{\displaystyle\sum_{\substack{\n_1:\partial\n_1=\emptyset \\ \n_2:\partial\n_2=\emptyset}} w(\n)}\, 
={\bf P}^{\emptyset ,\emptyset}_{\Lambda,\beta}[x\stackrel{\n_1+\n_2}{\longleftrightarrow} y]} 
\end{equation}
By implication, the presence of long range order  corresponds to percolation, i.e. the appearance of an infinite cluster,  in a duplicated system of sourceless random currents~\cite{Aiz82}.   To be precise, the natural definition of these terms refers to  infinite system, while \eqref{eq:xy} refer to finite ones, albeit arbitrarily large.  However, with some experience in the theory of percolation, the correspondence is easy to establish.  \\ 

A more explicit discussion of the extension of the RCR formalism to infinite systems can be found in \cite{AizDumSid15}.  There, RCR percolation arguments (uniqueness of the infinite cluster) are combined with bounds based on reflection positivity for a proof of  the continuity of the spontaneous magnetization in the n.n. models.   (Unlike the previous proofs of such statements, that argument covers also the $3D$ case).  \\[2ex]

\noindent  {\bf ii.  Cluster based hitting probability}\\ 

The following can be viewed as an RCR random cluster analog of the random walk's \emph{hitting probability} (for paths conditioned at two ends).  

  \begin{eqnarray}\label{xyz}
\frac{\langle\sigma_x \sigma_y\rangle_{\Lambda,\beta}\langle \sigma_y \sigma_z\rangle_{\Lambda,\beta}}{\langle \sigma_x \sigma_z\rangle_{\Lambda,\beta}}& =&  
\frac
{\displaystyle\sum_{\substack{\n_1:\partial\n_1=\{x,z\} \\ \n_2:\partial\n_2=\emptyset \quad }} w(\n_1) w(\n_2) \bm{1}[x\stackrel{\n_1+\n_2}{\longleftrightarrow} y] }
{\displaystyle\sum_{\substack{\n_1:\partial\n_1=\{x,z\} \\ \n_2:\partial\n_2=\emptyset \quad }} w(\n)}\, 
\\[2ex]   \notag  
&:=& \boxed{{\bf P}^{\{x,z\},\emptyset}_{\Lambda,\beta} [y \in  {\mathbf C}_{\n_1+\n_2}(x)]}\,. 
\end{eqnarray}
In the last expression the probability  refers to a pair of random currents with  $\partial (\n_1+\n_2) = \{x,z\}$.  
The corresponding multigraph includes a path linking the two points decorated by a collection of loops, which in the absence of percolation are all finite.  

The equality expresses the probability that the random cluster ${\mathbf C}_{\n_1+\n_2}$, which by the source constraints  includes a path linking $x$ with $z$, includes also the point $y$.   
While clusters are of a  more complicated geometry than  simple random paths, it is natural to expect the path analogy to be useful, and in particular in high dimensions.  \\ 
 
\noindent  {\bf iii. Covariance and the probability of avoided intersections}\\[2ex]  
Partial derivatives of correlations are expressed in truncated correlation functions,~e.g.
\begin{eqnarray} \mbox{ } \quad  \qquad 
\frac{1}{\beta} \frac{\partial}{\partial J_{u,v}} \langle\sigma_x \sigma_y \rangle  =  
\langle\sigma_x \sigma_y \sigma_u \sigma_v \rangle -
\langle\sigma_x \sigma_y \rangle \, \langle \sigma_u \sigma_v \rangle  
=:  
\langle\sigma_x \sigma_y \,;\, \sigma_u \sigma_v \rangle
\end{eqnarray}  
(omitting the common subscript $(\Lambda,\beta)$).   Such expressions appear in the 
derivative in $\beta$ of the energy density $\langle\sigma_x \sigma_y \rangle $ with $\{x,y\}$ a pair of nearest neighbors, and of  the magnetic susceptibility 
$
\chi(\beta) = \sum_{y} \langle\sigma_0 \sigma_y \rangle\,. 
$

The random current expression for such truncated correlations is
 \begin{eqnarray}\label{eq:trunc}
\langle\sigma_x \sigma_y \,;\, \sigma_u \sigma_v \rangle  &= 
\langle\sigma_x \sigma_u \rangle \, \langle \sigma_y \sigma_v \rangle 
{\bf P}^{\{x,u\},\{y, v\}}[u\stackrel{\n_1+\n_2}{\not\longleftrightarrow} v]  \notag \\  
& \, \,  +\, \, \langle\sigma_x \sigma_v \rangle \, \langle \sigma_y \sigma_u \rangle 
{\bf P}^{\{x,v\},\{y, u\}}[u\stackrel{\n_1+\n_2}{\not\longleftrightarrow} v] \,  .  
\end{eqnarray}
In  case $u$ and $v$ are neighbors and $x$ and $y$ are far away, the probabilities of no intersection 
refer to a pair of random clusters in which is imbedded  a pair of long paths  that start next to each other.

For a hint of what one may expect, let us note that for the simple random walks such probabilities are uniformly bounded away from $0$ only in  dimensions $d>4$.   This simple process can be viewed as describing critically damped random paths.  It is then natural to ask whether the same, or a close analog of that,  is true for the connected clusters associated with the  random current system at $\beta_c$ (which as we saw is its percolation threshold).  

As it turns out, a version of the $d>4$ part of the last assertion is provable, and it  plays a role in the proofs that the critical exponent bounds \eqref{MF_exp} are saturated in those dimensions.~\cite{Aiz82,AizFer86}.  \\

\subsection{The $4$-point truncated correlation and clusters' intersection probability} \mbox{  } \\ 

Since the main objective of these notes is the structure of the multi-spin correlations we now fast forward to the Ursell function $U_4$, which plays a key role in our discussion. \\  

Upon power series expansion, or equivalently differentiation,  of \eqref{U_def}  one finds that for sign-flip symmetric fields 
\begin{multline}  \label{U4_Ising}
 U_4(x_1,...,x_4) := \langle\sigma_{x_1} \sigma_{x_2} \sigma_ {x_3} \sigma_{x_4} \rangle - \\ 
- \left[ \langle\sigma_{x_1} \sigma_{x_2} \rangle \langle \sigma_ {x_3} \sigma_{x_4} \rangle  
+  \langle\sigma_{x_1} \sigma_ {x_3} \rangle \langle \sigma_{x_2} \sigma_{x_4} \rangle
+ \langle\sigma_{x_1} \sigma_{x_4} \rangle \langle \sigma_{x_2} \sigma_ {x_3} \rangle
\right]  
\end{multline}  
As a direct expression of the deviation of the $4$-point function's  from Gaussian behavior,  $|U_4(x_1,...,x_4)|$ should be measured against  $\langle\sigma_{x_1} \sigma_{x_2} \sigma_ {x_3} \sigma_{x_4} \rangle $, or  equivalently  the largest of the product of the two pair correlation functions. 
For that RCR yields the following pair of expressions. 
\begin{theorem}[\cite{Aiz82}] \label{thm:U4}For any ferromagnetic system of Ising spins, at $h=0$, 
\begin{multline}  \label{U4_RCR}
U_4(x_1,...,x_4) = \\ 
=  \ -2 \langle\sigma_{x_1} \sigma_{x_2} \sigma_ {x_3} \sigma_{x_4} \rangle   \,\, 
{\bf P}^{\{x_1,x_2,x_3,x_4\},\emptyset} 
[\mbox{all $x_j$ are $(\n_1+\n_2)$ interconnected}]\,  \\     
= \boxed{ -2 \langle\sigma_{x_1} \sigma_{x_2} \rangle \, \langle \sigma_ {x_3} \sigma_{x_4} \rangle   \,\, 
{\bf P}^{\{x_1,x_2\}, \{x_3,x_4\} } 
[\{x_1,x_2\}\stackrel{\n_1+\n_2}{\longleftrightarrow} \{x_3,x_4\}] }
\hspace{1.7cm}  
\end{multline}
\end{theorem}  
This important pair of identities are a simple consequence of Lemma~\ref{lem:switching}:  switching all  sources to $\n_1$, one finds 
that the resulting terms produce perfect cancellation if the $(\n_1+\n_2)$-connected cluster of $x_j$ form two distinct clusters.  The alternative is that all the sources are connected, and in that case the overall factors is $1-3=-2$.    From this observation follows the first line in \eqref{U4_RCR}. 
 The transition from the first expression to the second is by yet another application of the switching lemma. 
   
Let us add that the equality of the two expressions in \eqref{U4_RCR} is somewhat remarkable since the symmetry under permutations is obvious in one case but not the other.   In this respect,  the second relation is somewhat reminiscent of the (assumed) \emph{crossing symmetry}  of the conformal field theory, though here the relation is not limited to critical models.

We next turn to the implication of this relation for the model in different dimensions: $d=2$,  $d>4$, and then the marginal case of  $d=4$.   (One may note that we omit the case $d=3$, and not by chance.) \\ 

\bigskip 

\section{Emergent fermionic structures in $d=2$ dimensions} 

\subsection{Pfaffian correlations for boundary spins in planar models} \mbox{  }\\ 
\indent 
Since the breakthrough results of Onsager~\cite{Ons44}, and numerous works which followed, among the Ising models the planar ones form a special class, with strong links to other integrable systems.   Their  partition functions at $h=0$ admits a  determinant expression, in terms of $\{\beta J_{x,y}\}$-dependent matrices of size $2|\V|\times 2|\V|$.  For  translation invariant interactions on homogeneous graphs of degree $\nu$ the matrices can be presented as block diagonal with $4\times4$ blocks, indexed by the momentum (the Fourier transform variable), and through that a calculable expression follows. 

Furthermore, a rich collection of methods exist for  the lattice correlation functions~(cf. \cite{SchMatLie64,Kad69,KadCev71,CheSmir12,AizWar18} and references therein).  There is also a (partly separate) set 
of tools for the description of the conformally invariant field theory which is expected to emerge in the models' scaling limit~(cf. \cite{RivCar06,CamNew09,CheSmir12}).  

We shall not discuss here the special methods which have been devised for the analysis of the exactly solvable case of $2D$.  However, since in our discussion of the higher dimensions we shall invoke the random current representation,  let us note here some of its implications for the planar case.   
A striking example is the effectiveness of \eqref{U4_RCR} in explaining  the following curiosity, which 
was noted and derived by other means by Groeneveld-Boel-Kasteleyn~\cite{GBK78}.    

\begin{theorem}[Pfaffian structure of boundary spin correlations~\cite{GBK78}]\label{thm:pf_boundary}
For any Ising model on a planar graph $\G$, with only nearest-neighbor couplings $\{J_{x,y}\}$ (which need not be uniform), at $h=0$ and any $\beta$ the boundary spins' correlation functions  have the Pfaffian structure. 
\end{theorem}

\begin{figure}
\begin{center}
\includegraphics[width = 0.80\textwidth]{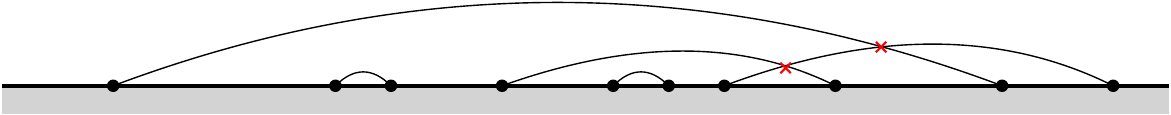}
\caption{A configuration of sites along the boundary of a planar domain.  
The  signature of each pairings' contributions to the Pfaffian equals the parity  of the number of crossings of the pairing paths ($(-1)^2 $ in the depicted example).}
\label{fig_pairing}
\end{center}
\end{figure}

 More explicitly:   
for any $2n$-tuple $(x_1,\dots,x_{2n})$ of different sites located along the outer boundary, or the boundary  of a fixed face of $\G$,  
\be \label{Pf_boundary}
\boxed{ \langle \sigma_{x_1}\cdots\sigma_{x_{2n}}\rangle_{\G,\beta} \ = \
\sum_{\pi \in \Pi_{2n}}  \sgn(\pi)  \prod_{j=1}^n \langle\sigma_{x_j} \sigma_{x_{\pi(j)}}\rangle_{\G,\beta}
\, }
\ee
where the sum is over pairing $\{j \leftrightarrow \pi(j)\}$
\footnote{Pairings satisfy  $\pi(j)\neq j$ and $\pi(\pi(j))=j$.  The signature is easily computed as the parity of the number of crossings in a generic collection of paths linking the corresponding pairs, if drawn on the same side of the graph's boundary, as depicted in Figure~\ref{fig_pairing}.  } 
 and $\sgn(\pi)$ is the parity of the number of pair crossings in any continuous deformation, along the boundary line, of the   labeled configuration under which  $\{x_j\}$  emerge relabeled as $\{x_{\pi(j)}\}$.

Equivalently, \eqref{Pf_boundary} can be expressed by saying that for the sites labeled in a cyclic order the correlation is given by a Pfaffian of the upper triangular array $[A_{i,j}]$ whose entries are given by the two point correlation function:\footnote{Pfaffians appeared very early in the exact solution of the free energy of two-dimensional Ising models \cite{Fis61,HurGre60,Kas63}, and are linked to determinants through the   relation $\Pf_n(A)^2  = \text{det}( A)$,
where $ A $ is the antisymmetric matrix whose entries   for $j<k$ are given by   $  A_{j,k} $.}
\be 
\langle \sigma_{x_1}\cdots\sigma_{x_{2n}}\rangle_{\G,\beta} = \Pf_n\left(\left[ \langle \sigma_{x_i} \sigma_{x_j}\rangle   \right]_{1\leq i<j\leq 2n}\right)\,.
\ee

Pfaffian structure of correlation functions is a fermionic counterpart to the Wick rule of the  (bosonic)  correlations of Gaussian fields.   In both cases all the $n$-point functions are simply expressible in terms of the two-point function.  Each evokes integrability in the sense of the existence of a non-interactive structure.  However, in contrast to the Gaussian correlations, Pfaffians indicate the presence of fermonic degrees of freedom.  
It is of interest to note that our RCR expression for $U_4$ sheds light on  {\bf  both cases}. \\     

Let $\{x_1,..., x_4\}$ be a cyclicly labeled configuration of points along the boundary of a finite planar graph $\G$.   In that situation any pair of  paths within   the   domain with intertwined end points ($x_1 \leftrightarrow x_3$ and $x_2 \leftrightarrow  x_4$), on its boundary have to intersect, and thus  
\be
{\bf P}^{\{x_1,x_3\}, \{x_2,x_4\} } 
[\{x_1,x_3\} \stackrel{\n_1+\n_2}{\longleftrightarrow} \{x_2,x_4\}] = 1 \,. 
\ee
Combined  with the last expression in \eqref{U4_RCR}  this implies 
\begin{align}  \label{U4_planar}
U_4(x_1,..., x_4)  = -2\langle\sigma_{x_1} \sigma_{x_3} \rangle \, \langle \sigma_ {x_2} \sigma_{x_4} \rangle\,, 
\end{align}  
which by the definition of $U_4$ means
\be  \label{Pf_4}
\langle\sigma_{x_1} \sigma_ {x_2}  \sigma_ {x_3} \sigma_{x_4} \rangle  = 
 \langle \sigma_ {x_1} \sigma_{x_2} \rangle - \langle\sigma_{x_1} \sigma_{x_3} \rangle \, \langle \sigma_ {x_2} 
\sigma_{x_4} \rangle +
\langle\sigma_{x_1} \sigma_{x_4} \rangle \, \langle \sigma_ {x_2} \sigma_{x_3} \rangle\,.
\ee
This proves \eqref{Pf_boundary} for the  $4$-point function.  The argument's extension to  higher correlations can be found in \cite{ADW}. \\  

\subsection{Emergent planarity in $2D$ finite range models} \mbox{  } \\

The above strict relations do not carry over into the extensions of the ferromagnetic model to planar models with more general variables in the Griffiths-Simon class.  The reason for that is also clear from the above argument:    Pfaffian correlations for systems built from Ising spins require probability $1$ for the actual intersection of paths whose planar projections cross.  This feature is lost in the construction depicted in Figure~\ref{fig_blocks}.  

Nevertheless, taking into account the expected fractal nature of the relevant RCR clusters at $\beta_c$, it is natural to conjecture that the Pfaffian structure of correlations emerges at the models' critical points.  
A related  situation shows up in $2D$  models on $\Z^2$ with finite range interactions which are not planar (and thus also no longer solvable).  For this case the following result  was proven  using the RCR and percolation type techniques.

\begin{theorem}\label{thm:Pf_finite_range}[\cite{AizDumTasWar18}]
Let $J$ be a set of coupling constants for an Ising model over the restriction of $\Z^2$ to the upper half-plane   $\bbH$, 
 which are:  
\begin{itemize}[noitemsep]
\item[(i)] ferromagnetic, $J_{x,y}\ge0$ for every $ x,y\in\mathbb \Z^2\cap \bbH$, 

\item[(ii)]  finite-range  and such that the induced graph is connected, 

\item[(iii)]  translation invariant, i.e.~that $J_{x,y}=J(x-y)$, and 
 
\item[(iv)] reflection invariant: $J_{0,(a,b)}=J_{0,(-a,b)}=J_{0,(a,-b)}=J_{0,(b,a)}$ for every $a,b\in \bbZ$.   \end{itemize}
  Then, for any collection of boundary points  $x_1=(k_1,0),\dots,x_{2n}=(k_{2n},0)$  satisfying $k_1<k_2<\dots<k_{2n}$ 
 \be \label{Pf_cor}
 \boxed{ \langle \sigma_{x_1}\cdots\sigma_{x_{2n}}\rangle_{\beta_c} \ = \
\Pf_n \big( \big[\langle \sigma_{x_i}\sigma_{x_j}\rangle_{\beta_c}\big]_{1\le i<j\le 2n} \big)\big[1 +o(1)\big]
\, } 
\ee
where $\beta_c$ is the critical inverse-temperature of the model, and $o(1)$ is a function of the points $x_1,\dots,x_{2n}$   which tends to zero for configuration sequences with  $\min_{1\le i<j\le 2n} \{|x_i-x_j|\} \to \infty$. 
\end{theorem}

An intuitive explanation of the mechanism at work here is the 
fractality of the connecting paths, and of the connected clusters, at the critical point.   
Even if, due to the finite range of hopping the connections over a two dimensional graph are not strictly planar, when the projections of two fractal paths cross, there would typically  be multiple switchbacks, in which case non-intersection would require a large number of avoided linkages. The probability that all such opportunities to connect would be missed may be expected to vanish exponentially in the number of such encounters.  By scaling, that number should typically grow at least logarithmicaly in the distance  (with a related estimate for the probability of  exceptional behavior). 

As this outline suggest, the proof of Theorem~\ref{thm:Pf_finite_range}  is seeped in percolation type arguments, spelled in~\cite{AizDumTasWar18}.  
The stochastic geometric picture suggests that also the following may be true.  

\begin{conjecture} 
The conclusion of Theorem~\ref{thm:Pf_finite_range}  holds also for similar systems with the Ising spins replaced by spin variables in the Simon-Griffiths class (at the correspondingly adjusted $\beta_c$).  
\end{conjecture}

With only some coincidental exceptions~\cite{AizWar18},  Ising spin systems which are not strictly planar are beyond the reach of exact solutions.   However there has been recent progress on rigorous perturbative methods which allow to prove a degree of universality in certain aspects of the critical behavior, of models close enough to the solvable case~\cite{GiuMas13,Giu21}.   In comparison, Theorem~\ref{thm:Pf_finite_range}  provides less detailed information, but it is is not limited to weak couplings.    

\subsection{Non-gaussian structure of the scaling limit of $2D$ bulk correlations}  \mbox{  }  \\  

The RCR explanation of the Pfaffian structure of the boundary-to-boundary spin correlations makes it also clear that this behavior does not extend to the multi-spin correlations in the bulk.   

It was also noted in 
\cite{Aiz82} that by elementary geometric considerations in the symmetric case of the vertices of a square allied with the symmetry axes,   for two random currents whose  sources at the pairs of opposite vertices,   the probability of paths intersection is at least $1/2$, and thus for any $\beta$: 
\be 
\boxed{ |U_4(x_1,...,x_4)| \ \geq \ \langle \sigma_{x_1} \sigma_{x_3}\rangle \, \langle \sigma_{x_2} \sigma_{x_4}\rangle } 
\ee 
 Related considerations allow to deduce that the $2D$ model's scaling limit is definitely not gaussian.  However, since this model is solvable by other means, we should not delve into this point here.  \\

\section{The $U_4$ gaussianity criterion} \label{sec:U4criterion} 

The non-vanishing in the scaling limit of the properly rescaled $4$-point function $U_4$ clearly implies that the limit is not Gaussian.  The converse statement is not obvious, but  for the models considered here is also true.  This fact was first noted and proven by C. Newman~\cite{New75}(Theorem 8) through the  Lee-Yang property.

\begin{proposition}[\cite{New75}]
If, for an Ising model,  $U_{2m}(x_1,..., x_{2n}) \equiv 0$ for some $m\in \mathbb N$,   then the model is Gaussian. 
\end{proposition} 
More quantitative estimates can be derived from the expressions provided in \cite{New75} for the Ursel functions in terms of the Lee-Yang zeros.  Alternatively, the statement can also be deduced from the following bound on the $N$-point functions' deviation from the corresponding Week rule, which was derived in \cite{Aiz82}  through the random currents' switching lemma.   

\begin{proposition}[\cite{Aiz82}(Theorem 12.1)] \label{prop:U4_suffices}
In ferromagnetic Ising models,  
for any $n\in \mathbb {N}$
\begin{multline}
 \label{U4_suffices}
0  \leq \ \mathcal {G}_{2n}[S_2](x_1,\dots,x_4) - S_{2n}(x_1,...,x_{2n})  \\ 
\leq  \ \boxed{\,  \frac{-3}{2} \sum_{1\leq j<k<l<m \leq 2n}   U_4(x_j,x_k,x_l,x_m) \cdot 
\mathcal {G}_{2n-4}[S_2](...,\cancel{x_j},...,\cancel{x_k},...,\cancel{x_l},...,\cancel{x_m},...  )}
\end{multline} 
where $\mathcal {G}_{2n}$ is the Wick functional defined in \eqref{triv} 
\end{proposition}  

One should note that \eqref{U4_suffices} is homogeneous of degree $1$ in each of the spin variables, and hence the relation passes onto the scaling limit.  It follows that the limit is Gaussian if and only if the ratio of $|U_4(x_1,...,x_4)|$ to $  S_4(x_1,...,x_4)$ tends to zero, i.e. the Wick law holds at the level of the $4$-point function.

It follows that a quantifier of the deviation from gaussianity can be found in the dimensionless ratio 
\be \label{U4ratio} \notag
\frac{U_4(x_1,...,x_4)  }{ \langle\sigma_{x_1} \sigma_{x_2} \sigma_{x_3} \sigma_{x_4} \rangle}
\ee
for quadruples of sites at pairwise distances of order $L$. 
For concreteness sake, one may represent this through the scale dependent ratio
\be
\boxed{R_L(\beta) := \frac{\displaystyle \sum_{x_1,...,x_4 \in \Lambda(L)}|U_4(x_1,...,x_4) |} 
                 { \Big[ \displaystyle \sum_{x_1,x_2 \in \Lambda(L)}S_2(x_1,x_2) \Big]^2}  }
                \left.\begin{cases} \leq \\ \geq \frac{1}{3} \end{cases} \!\!\!\!\!\! \right\}
           \frac{\displaystyle\sum_{x_1,...,x_4 \in \Lambda(L)}|U_4(x_1,...,x_4) |} 
                 { \displaystyle\sum_{x_1,...,x_4 \in \Lambda(L)}S_4(x_1,...,x_4) }  
                 \ee 
where the inequalities are stated just to convey the essential equivalence of the two slightly different ratios. 

With the above choice, $R_L$ can be presented as minus the Ursel function  
\be
R_L(\beta) = -\left[ \langle T_L^4 \rangle_{\beta} - 3 \langle T_L^2 \rangle_{\beta} ^2 \right] \,. 
\ee 
associated with the block-average of $\sigma_x$:  
\be \label{def_Tf_scaled_sigma}
T_{L}(\sigma) :=  \frac{1}{\sqrt{\Sigma_L}}\sum_{x\in \Lambda(L)}  \, \sigma_x   \, 
\ee
which by its normalization is of unit variance:  $\langle T_L^2 \rangle_{\beta}  =1$. 

Proposition~\ref{prop:U4_suffices} has the following  notable implication.
\begin{corollary}  \label{Corro_U4}
In any sequential limit in which $\beta \nearrow \beta_c$, $L\to \infty$ 
and
\be \label{R_condition} 
 R_{rL}(\beta) \to 0\,  \quad \mbox{for  each $r<\infty$} \, 
\ee  
the collection of variables  $\{T_{f,L} :  f\in C_0(\R^d)\}$  
 of \eqref{def_Tf_scaled}  converges in distribution, and their joint limit describes a Gaussian field. 
 \end{corollary}
 \begin{proof} By a standard general argument, to establish the joint convergence of a collection of variables to a Gaussian limit it sufficed to prove convergence to a Gaussian limit for each of the variables' linear combinations.   We establish that by focusing on the  moment generating functions of $T_{f,L}$, with in the limit $L\to \infty$ carried at fixed $f$.   
 
 For  convenience, we present the argument first for $T_L$ of  \eqref{def_Tf_scaled_sigma}.  (The corresponding $f$, being an indicator function is not in $C_0((\R^4)$, but continuity is not essential for the argument.) \\    
 
 i)  
Applying \eqref{U4_suffices}  to $\langle T_L^{2n} \rangle $  we learn, through elementary combinatorics, that for any $n\geq 2$
\be
0 \  \leq \  \frac{(2n)!}{2^n n!} \langle T_L^2 \rangle^n - \langle T_L^{2n} \rangle \ \leq \ 
\frac{3}{2}  { 2n \choose 4}\  R_L  \   \langle T^{2n-4} \rangle 
\ee
Upon summation over $n\geq 2$  with weights $z^{2n}/(2n)!$ (and noting that also the middle term vanishes for $n=0,1$) this yields  
\be
0 \ \leq \   \exp  \big\{\frac{z^2}{2} \langle T_L^2 \rangle  \big \} - \langle \exp\big\{  z T_L \big \}   \rangle  \  
\leq  \   \frac{z^4}{2^4}   \  R_L  \  \langle \exp\big\{  z T_L \big \}   \rangle  \,,  
\ee
Recalling that $\langle T_L^2 \rangle=1$ one gets, under the stated assumption,
\be
\boxed{ \Big|  \langle \exp\big\{  z T_L \big \}   \rangle - \exp \{  z^2/2   \}   \Big |\  
\leq  \      2^{-4} \ z^4   \exp\{  z^2/2  \}   \cdot  \  R_L  \   \to 0  }
\ee
from which the first claim follow.  \\ 

\noindent  ii) For the more general variables $T_{f,L}$, of \eqref{def_Tf_scaled}, 
the  moment analysis  yields
\be \label{moment_Ising}
\boxed{ \Big|  \langle \exp\big\{  z T_{f,L} \big \}   \rangle - \exp  \big\{\frac{z^2}{2} \langle T_{f,L}^2 \rangle  \big \}   \Big |\  
\leq  \      2^{-4} \ z^4   \exp  \big\{\frac{z^2}{2} \langle T_{|f|,L}^2 \rangle  \big \}    \cdot  \  \widetilde R_{|f|,L}  \  } 
\ee
with 
\be 
\widetilde  R_{f,L}(\beta) := \frac{\sum_{x_1,...,x_4 \in \Lambda(L)}|U_4^{(\beta)}(x_1,...,x_4) |  \prod_{j=1}^4  f(x_j/L)}
                 { \left[ \sum_{x_1,x_2 \in \Lambda(L)}S_2^{(\beta)}(x_1,x_2) \right]^2}
                 \leq r^d \|f\|_\infty^4 R_{rL}(\beta)
 \ee 
and 
\be 
\langle T_{f,L}^2 \rangle = \frac{\sum_{x,y\in \Lambda(rL)}  \langle \sigma_x   \sigma_y \rangle   \, f(x/L) \,  f(y/L) }
 {\sum_{x,y\in \Lambda(L)}   \langle \sigma_x   \sigma_y \rangle  } \leq  r^d  \|f\|_\infty^2 \,. 
\ee 
The inequalities in the last two expressions are based on the  Messager  Miracle-Sol\'e monotonicity of the two point function~\cite{MesMir77}, as explained in \cite{AizDum21}.     
Taken together they imply convergence of the moment distribution function to a finite limit  which coincides with that of a Gaussian random variable of variance 
\be 
\lim_{L\to \infty}  \langle T_{f,L}^2 \rangle = \frac{\int_{\Lambda(rL)^2}  S_2(x, y)     \, f(x) \,  f(y) \, dx\, dy}
 {\int_{\Lambda(rL)^2}  S_2(x, y)  \, dx\, dy} \, .
 \ee 

The stated distributional convergence for $T_{f,L}$ follows, as it did for $T_{L}$. 
\end{proof} 

\section{Gaussianity of the scaling limits in $d>4$ dimensions}

\subsection{Sketch of the proof of Theorem~\ref{thm:gaussian phi4} for the Ising case}\mbox{  } \\ 
   
The  exact expression in \eqref{U4_RCR} leads to the following pair of upper bounds~\cite{Aiz82} . 
\begin{lemma}  For the Ising model on any finite graph, at $h=0$ and $\beta\geq 0$ (which are omitted below) 
 \begin{multline}  \label{eq:intera} |U_4(x_1,...,x_4)|      
 \leq   2\ \langle \sigma_{x_1}\sigma_{x_2}\rangle \langle \sigma_{x_3}\sigma_{x_4}\rangle  \  \times \\  
\times  {\bf P}^{\{x_1, x_2\}, \{x_3, x_4\},\emptyset,\emptyset} [ {\mathbf C}_{\n_1+\n_3}(x_1)\cap{\mathbf C}_{\n_2+\n_4}(x_3)\ne \emptyset] 
  \end{multline} 
where the probability refers to four  {\bf independently distributed} random currents with the prescribed sources: 
$$\partial (\n_1) =\{x_1,x_2\}\quad,  \partial (\n_2) =\{x_3,x_4\}, \quad \partial (\n_3)=\partial(\n_4) =\emptyset \,.$$   
And [consequently] the following diagrammatic bound holds
 \be  \label{eq:interb}
\boxed{   |U_4(x_1,...,x_4)| \,  \le \,    2 \,   \sum_{u} \langle \sigma_u\sigma_{x_1}\rangle  
 \langle \sigma_u\sigma_{x_2}\rangle  
  \langle \sigma_u\sigma_{x_3}\rangle      
   \langle \sigma_u\sigma_{x_4}\rangle 
   }
\ee   
\end{lemma}

Inequality \eqref{eq:intera}  is established through analysis which we omit here (cf. \cite{Aiz82}).   However the transition to 
\eqref{eq:interb}  is then elementary.   It amounts to a bound on the probability of intersection by the expectation value of its size (number of points),
\begin{multline}
{\bf P}^{\{x_1, x_2\}, \{x_3, x_4\},\emptyset,\emptyset} [ {\mathbf C}_{\n_1+\n_3}(x_1)\cap{\mathbf C}_{\n_2+\n_4}(x_3)\ne \emptyset]  \ \leq \ \\ 
\leq \,  {\bf E}^{\{x_1, x_2\}, \{x_3, x_4\},\emptyset,\emptyset} [ \, \, |{\mathbf C}_{\n_1+\n_3}(x_1)\cap{\mathbf C}_{\n_2+\n_4}(x_3)| \, \, ]  \,  \ = \ \\ 
= \sum_u   {\bf P}^{\{x_1, x_2\}, \{x_3, x_4\},\emptyset,\emptyset} [ u\in {\mathbf C}_{\n_1+\n_3}(x_1)\cap{\mathbf C}_{\n_2+\n_4}(x_3)]  
\hspace{1.2cm}  \mbox{  } 
\end{multline} 
Combined with  \eqref{xyz}, which provides an exact expression for the last term, one obtains the \emph{digrammatic} upper bound stated in ~\eqref{eq:interb}.

The gaussianity of the Ising model's scaling limits for $d>4$ was  established through the combination of the {\em tree diagram bound} \eqref{eq:interb} of~\cite{Aiz82}
and the \emph{infrared} upper bound (for the nearest neighbor  interaction of constant strength $J$)
\be \label{infrared}
\boxed{ \forall  \beta \leq \beta_c: \quad \beta J \langle \sigma_x \sigma_y\rangle_{\beta}  \leq  \frac{C}{|x-y|^{d-2}}\,}
\ee  
which was presented above in \eqref{crit_decay}.  Omitting some details, which are covered in the original references, the combination of the two bounds yields for the nearest neighbor Ising model the estimate 
\be \label{IR_bound}
\boxed{ R_L(\beta) \leq  \ C \frac{L^d (\Sigma_L)^4} {(\Sigma_L)^2}  \leq  \frac{C}{L^{d-4}} } 
\ee  
which holds with a uniform constant for all $\beta<\beta_c$ and $L<\infty$.  
The claimed results then follows by a direct application of the criterion stated in Corollarly~\ref{Corro_U4}.

\subsection{Sketch of the proof of Theorem~\ref{thm:gaussian phi4} for $\varphi^4$ variables} \label{sec_extension to GS}\mbox{  } \\ 
 
The proof of the $\varphi^4$ version of Theorem~\ref{thm:gaussian phi4} in \cite{Aiz82} proceeded through an extension of the above analysis to a broader class of single spin distributions.   This is enabled by the observation, of Griffiths and Simon~\cite{SimGri73}) that the $\varphi^4$ single spin distribution can be presented as the limiting distribution of weighted averages of ferromagnetically coupled Ising spins. 
The  explicit definition of the GS class of variables is given next, but is not essential for the following summary.

By the linearity of the relation between the Ising spins and their aggregates $\varphi$, statements that were derived for the Ising spin model through relations which are homogeneous of degree $1$  in the spin variables apply also to systems based on $\varphi^4$ variables (all that is needed is the substitution  $\sigma_x \Rightarrow \varphi_x$).  That applies in particular to \eqref{U4_suffices} of Proposition~\ref{prop:U4_suffices}   Hence  the entire discussion of Section~\ref{sec:U4criterion} applies mutatis mutants also to systems based on $\varphi^4$ variables.  

At first sight, the above rule does not apply to the tree diagram bound \eqref{eq:interb}, where the left side is homogeneous of degree $4$ in $\sigma$, while the right is of degree $6$.  This was not an issue for Ising spins since there $\sigma^2\equiv 1$.   However one may proceed by  first transforming the relation into a dimensionally balanced one.   From this perspective, $\sigma^2$ is not dimensionless but    $\beta J_{u,v} \sigma_u \sigma_v$ is.

Dimensional deficiency in Ising relations can in many cases be corrected through ``point splitting''.   A case in point is the following dimensionally balanced version of the diagrammatic inequality \eqref{eq:interb}.

\begin{proposition} [\cite{AizGra83}] \label{prop:tree_phi} For any ferromagnetic system of Ising spins with the Hamiltonian \eqref{H_phi} at $h=0$, at each $\beta\geq0$ (whose presence as a subscript on $U$ and $S$ is omitted below) 
 \begin{align} \label{eq:inter_phi} \notag
 \boxed{|U_4(x_1,...,x_4)|  
   \leq    2 \!  \! \!   \! \! \sum_{\substack{y \notin \{x_1,..., x_4\}\\ u,v}}   \!  \! \!   \!  
   S_2(x_1,y) \, 
   S_2(x_2,y) \,  
  [\beta J_{y,u}]   S_2(u,x_3) \,    \,    [\beta J_{y,v}]   S_2(v,x_4) }  \notag \\   
   \hspace{-0.6cm} \boxed{ + 2 \left[ \sum_{u\neq x_1}  S_2(x_1,x_2) S_2(x_1,x_3) [\beta   J_{x_1,u}] S_2(u,x_4)   + \mbox{$3$ permutations} \right] }
  \end{align} 
  where $S_2(x,y)= \langle \sigma_{x}\sigma_{y}\rangle_\beta$.   
\end{proposition}

This version of the tree diagram bound extends verbatim to systems with variables in the Griffiths Simon class, and in particular the lattice field variables with the $\varphi^4$-type a priori measure.  The only change is to switch the notation to $S_2(x,y)= \langle \varphi_{x}\varphi_{y}\rangle_\beta$, 
and make the similar adjustment in $U_4$.   (This will be an obvious consequence of the deconstruction of the $\varphi$ variables into Ising spins which we postponed to section~\ref{sec_GS}.)

 The other essential ingredient for the Ising analysis was the \emph{Infrared Bound} \eqref{IR_bound}.  However already the original derivation (in its different forms~\cite{GliJaf73,FroSimSpe76,Sok82}) has been directly valid for all single site distributions of sub-exponential decay.  
 
By the above means, the proof of the gaussianity of the Ising model's scaling limits extends to the generality in which it is stated in Theorem~\ref{thm:gaussian phi4}.  \\

\section{A decomposition of $\varphi^4$ variables into Ising constituents}  \label{sec_GS}

  \subsection{The Griffiths-Simon class of scalar variables} \label{sec:GS}\mbox{  }\\ 
  
Through some minor adjustments, the results presented above for the ferromagnetic Ising model extend to similar lattice systems with various other single variable distributions, including the $\varphi^4$ fields.   A device which was designed for such a purpose is the Simon-Griffiths representation of $\varphi^4$ variables as the distributional limit of  block averages of ferromagnetically coupled Ising ``elementary'' constituents.   

This technique was used by Griffiths~\cite{Gri69} for an extension of the Lee-Yang theorem, as well as his eponymous inequalities, to a broader range of spin distributions.  Among the well recognized examples in that class are the ``classical'' bounded spins, with uniform distribution in $[-S,S]$,  and their ``quantized'' analogs in which the  spins take values in $[-S,S] \cap \Z$ with equal a-priori weights.  The construction was further extended by B. Simon and R. Griffiths to cover also the $\varphi^4$ distribution~\cite{SimGri73} which is of interest here.   Following is its more explicit formulation. 

\begin{figure}
\begin{center}
\includegraphics[width = 0.40\textwidth]{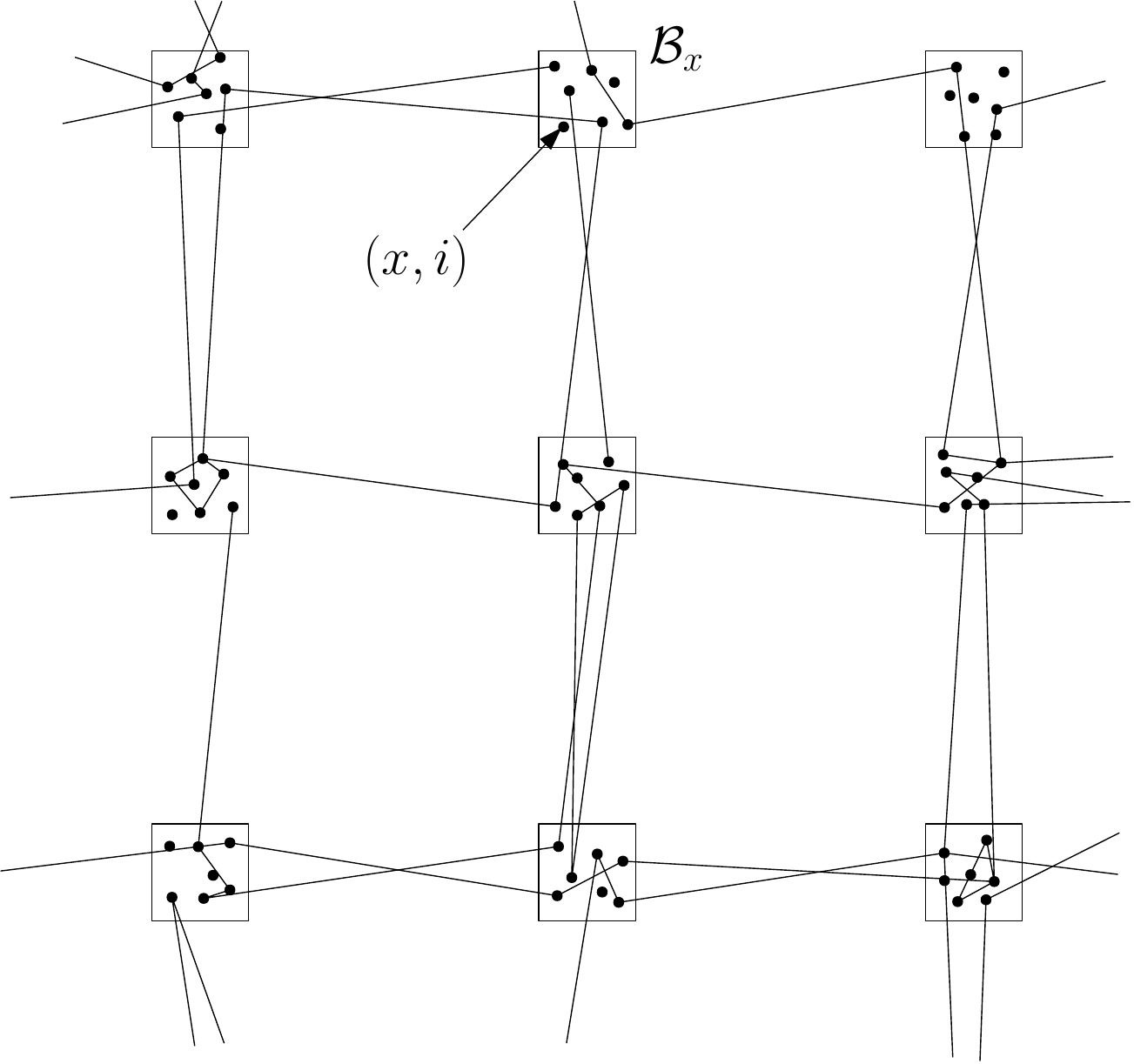}
\caption{The decorated graph, in which the sites $x$  of a graph of interest are replaced by ``blocks'' 
$\calB_x$ of sites indexed as $(x,n)$.   The Ising ``constituent spins'' $\sigma_{x,n}$ are coupled pairwise through intra-block  couplings $ \delta_{x,y} K_{n,m}$ and inter-block couplings $J_{x,y}$.  The  lines indicate a  possible realization of the corresponding random current.}
\label{fig_blocks}
\end{center}
\end{figure}
 
\begin{definition} \label{def:rho}
A probability measure on  $\rho $ on $\R$ is said to belong to the Griffiths-Simon (GS) class if either  of the following conditions is satisfied 
\\  
\indent $1)$   the expectation values with respect to $\rho$ can be presented as 
\be  \label{Griff_cond} 
\int F(\varphi) \, \rho(d\varphi)  =   \frac{1}{\rm{norm}}  \sum_{\underline{\sigma}\in \{-1,1\}^N} F(\sum_{j=1}^N \alpha_j \sigma_j)
\,\, e^{\sum_{i,j=1}^N  K_{i,j} \alpha_i \alpha_j  \sigma_i  \sigma_j}  
\ee 
with some $\{b_j\}\subset \R$, and $K_{i,j}  \geq 0$.   \\  
\indent $2) $   $\rho$ can be presented as a (weak) limit of probability measures of  the above type, and is of sub-gaussian growth: 
\be \label{sub_gauss}
\int e^{|\varphi |^\alpha} \rho(d \varphi ) < \infty \quad \mbox{for some $\alpha> 2$.}
\ee 
A random variable is said to be of  Griffiths-Simon type if its probability distribution is in the GS class.   
\end{definition} 

The Griffiths condition \eqref{Griff_cond}  can be equivalently expressed by  saying that $\phi$ can be presented as  
\be \label{decomposition} 
\phi = \sum_j \alpha_j \sigma_j \,. 
\ee 
The ferromagnetic system of such variables  $\{\varphi_x\}$ with the Hamiltonian  
\be \label{H_phi}
{H}(\varphi) :=  -  \sum_{\{x,y\} \subset \Z^d} J_{x,y} \varphi_x \varphi_y - h   \sum_{x \in \Z^d} \varphi_x 
\ee
 can be presented as an Ising spin system, with the interaction
\begin{eqnarray}  \label{H_I_gen}
\mbox{  } \hspace{1.3cm} \hfill \boxed{ H + K = 
-\!\!\!\!\!\!\sum_{\substack{\{x,y\} \in \E\\ i,j\in \{1,...,N\} }}   \alpha_i \alpha_j  J_{x,y} \sigma_{x,i}  \sigma_{y,j}  - h     \sum_{x,j} \alpha_j  \sigma_{x,j} + K(\sigma)  }
\,. 
\end{eqnarray}
where $K(\sigma)  =- \sum_x \sum_{i,j} K_{i,j} \alpha_i \alpha_j  \sigma_i  \sigma_j $.  \\ 

The infinite volume Gibbs equilibrium states of such systems are then constructible as a double limit, in which one first takes $N\to \infty$, to produce the desired convergence of the distribution of the $\varphi$ variables, followed by  the infinite volume limit $\mathcal L\to \infty$.  \\

Of particular relevance to us is the observation of  \cite{SimGri73} that a variable $\{\varphi_x\}$  with the probability distribution 
$$\rho_{\lambda,b}(d\varphi) = e^{-\lambda \varphi^4 + b\varphi^2} d\varphi/\text{norm} $$   
can be produced as the $N\to \infty$ distributional limit   of the collection of the block averages of elemental Ising spins $\{\sigma_{x,j}\}$ (the dots  in Fig.~\ref{fig_blocks})
\be \label{phi4_Ising}
\varphi_x^{(N)} =   \frac{\alpha(N)}{\sqrt N}\sum_{j=1}^N \sigma_{x,j}
\ee 
which interacted through the ``ultra-local'' on-site coupling
\be
K(\sigma) = -   \frac{g(N)}{N } 
\left[ \sum_{j=1}^N \sigma_{x,j} \right]^2 
\ee 
at  suitably adjusted values for $\{\alpha(N), g(N)\}$.   The explicit values of the required  adjustments in the parameters of $K$ are easy to determine, but are not relevant for the present discussion, beyond the fact of feasibility of this construction for any $\lambda \geq 0$ and $b$ (of arbitrary sign). \\  

Underlying the last assertion is the observation that in this case $K$ is a mean-field interaction of $N$ spins.  In the limit $N\to \infty$ it undergoes a phase transition at $g=1$, which is explicitly analyzable.   Expanding the well known expression for the entropy as function of $\frac1 N \sum_j \sigma_{x,j}$, to   the fourth order  (around its maximum),   one finds that at the mean-field critical point the leading $\varphi^2$ term vanishes, and $\varphi^4$ emerges.   With suitable fine-tuning of $\{\alpha(N), g(N)\}$  the distribution of $\varphi_x^{(N)}$ can be made to converge to any $\rho_{\lambda,b}$ at $\lambda>0$.

To summarize:  through the above construction, any system  
of $\varphi^4$ variables  associated with the sites of a finite graph $\G$, and coupled through the graph's edges, is presentable as the limit ($N\to \infty$) of a system of constituent Ising spins associated withe the  decorated graph produced by replacing each vertex of $\G$ by the the complete graph of $N$ vertices.  In it, the ``microscopic sites'' are indexed by $\{x,j\}$ and two sites are neighbors if they either lie in a common block (i.e. same $x$) or in a pair on neighboring blocks (as depicted in Fig. \ref{fig_blocks}).

\subsection{Dimension balancing through point splitting}  \label{sec:point_splitting}\mbox{  }  \\ 

The linearity of the  decomposition \eqref{decomposition} allows to extend to systems with  variables in the GS  class 
many of the results that were originally derived for ferromagnetic Ising spin models. 
In particular, a suitably weighted version of \eqref{eq:intera}  holds,   allowing to assign a probabilistic interpretation to the ratio  
$$ U_4(x_1,...,x_4)/S_4(x_1,...,x_4)$$ 
also for the field variables.  
However the intersection event needs to be interpreted now in terms of the connected clusters of the elemental constituents of the $\varphi$ variables  (the small dots on Figure~\ref{fig_blocks}).   One may quickly spot 
two problems, which are actually related:
\begin{enumerate} 
\item[i.]  bounding the probability that two clusters intersect by the size of their intersection on the enhanced graph may lead to a more severe over-estimate of $U_4$ than the one which was discussed above for the Ising model,
\item[ii.]   without the restriction to $\sigma^2=1$, the diagrammatic bound \eqref{eq:interb}  is unbalanced and hence cannot be valid  as stated for the more general spins  
\end{enumerate} 
To reduce the over-counting, one may bound the probability of intersection through the number of the sites of the graph $\V$, i.e. the boxes, rather than the ``internal dots'' in the decorated graph.   

Furthermore, for the above purpose, the probability that a current with specified sources reaches a box can be bounded above through the mean  flux into the box through edges connected to the sources.   This results in the following  bound (Proposition 7.1 in~\cite{AizGra83}).   As an added benefit, this provides a dimensionally balanced substitute for the site hitting relation \eqref{xyz}: 

\begin{lemma}[Box hitting probability \cite{AizGra83}]   Under the assumptions of Proposition~\ref{prop:tree_phi}, for any subset $\calB \subset \V$, and pair of sites $\{x,y\}$ in    $\calB^c:=\V\setminus \calB$
\begin{align}\label{xy_B}
\boxed{ \langle \sigma_x \sigma_z\rangle  \, 
{\bf P}^{\{x,z\},\emptyset}[ {\mathbf C}_{\n_1+\n_2}(x) \cap \calB \neq \emptyset] \ \leq \!\!
\sum_{ u\in \calB, \,\,  v\in \calB^c}  \!\!\!
\langle \sigma_x \sigma_u\rangle \ [\beta   J_{u,v}]  \, \langle \sigma_v\sigma_y\rangle }
\end{align}
\end{lemma}  

In this inequality the right hand side is an estimate on the double current's flux into $B$ through edges which in $B$'s complement are $\n_1+\n_2$ connected to $x$. 
\footnote{To avoid confusion let us recall that random current's flux is directionless, and void of a meaningful sign.  The current's ``conservation law'' involves only the flux's parity.} 
 The proof of this bound is an exercise in geometric combinatorics based on the switching lemma.  The argument, which can be found in the provided reference, can be simplified a bit for the $x\Leftrightarrow y$ symmetrized version of \eqref{xy_B}, which  for practical purposes is as useful.. 

With  \eqref{xy_B} applied in place of \eqref{xyz} in the derivation of  the diagrammatic bound \eqref{eq:interb},  one gets the dimensionally balanced inequality \eqref{eq:inter_phi}.  As explained above, this allows to extend the proof of the gaussianity of the Ising model's scaling limits in $d>4$ dimensions to variables in the GS class, including the $\varphi^4$ case.

\section{Gaussianity at the marginal dimension $d=4$}

\subsection{A suggestive dimensional analysis} \mbox{  }  \\  

The above sketched proof of the scaling limits'  gaussianity  in $d>4$ dimensions may be recapitulated in a suggestive  dimensional argument, which helps to frame the next question.

The proof's  key input consists of the {\em tree diagram bound} \eqref{eq:interb} 
and the \emph{infrared upper bound} \eqref{crit_decay} 
To focus on the essence,  one may initially proceed under the (reasonable) assumption that: i)  at $\beta_c$ the two-point function is of comparable values for pairs of sites at similar distances, ii) 
on the scales under consideration $S_2(x,y)$ decays by a power law.\footnote{The  behavior assumed in \eqref{eta} does not hold for  $\beta \neq \beta_c$ at distances much larger than the correlation length, but is a common feature of critical models with power law decay of correlations.}   
In the standard notation that would be presented as 
\be\label{eta}
\beta_c J \langle \sigma_x \sigma_y\rangle_{\beta_c}  \approx  \frac{C}{|x-y|^{d-2+\eta}}\, , 
\ee   
for which \eqref{infrared} implies $\boxed{\eta\geq 0}$.
  
Then, for  quadruples of points at mutual distances of order $L$, in lattice units (!), the sum in the tree diagram bound \eqref{eq:interb}  contributes a factor $L^d$ while the summand has two extra correlation function factors $S_2(\cdot,\cdot)$, in comparison to $S_4(x,y,z,t)$.  
each $S_2$ factor being of order $1/L^{d-2+\eta}$.  

This suggests that for any collection of four point which are at mutual distances of order $L$
\be \label{D>4} 
 \boxed{ \frac{|U_4^{(\beta)}(x_1,...,x_4) |} 
                 { S_4^{(\beta)}(x_1,...,x_4) }
         \leq C \, \frac{L^d}{L^{2(d-2+\eta)}} = \frac{C}{L^{d-4+2\eta}  }}
\ee
which for $d>4$   vanishes in the limit $L\to \infty$.   \\  

Up to technical details, the above captures the essence of the argument underlying the gaussianity of the Ising model's scaling limit in \cite{Aiz82}.   
A similar bound for $\varphi^4$ models was derived also in \cite{Fro82}, by other means.     \\  

The above dimensional argument may  appear similar to Wilson's heuristic renormalization group analysis~\cite{Wil71}.    
However it should be stressed that: 
\begin{enumerate} 
\item[i.]   the bounds on $U_4$ of \cite{Aiz82, Fro82} are expressed in terms of relations among the ``dressed'' quantities, not just their initial values  near a gaussian starting point, 
\item[ii.]  in general, the two point function's critical behavior is corrected by the critical exponent $\eta$, and hence the dimensional analysis relies on the non-perturbative fact (implied by a reflection positivity argument) that in this model $\eta\geq 0$.  \\ 
\end{enumerate}        
\mbox{  } 

\subsection{The challenge} \mbox{  } \\ 

The initial proofs of gaussianity in high dimensions ($d>4$)~\cite{Aiz82,Fro82}, as well as the heuristic outline given above,  do not extend to $d=4$ (unless $\eta>0$ which however does not seem to be the case).   More explicitly, if at this marginal dimension $\eta=0$ (as expected),   then  the tree diagram bound on $R_L$ no longer  decays to $0$ for  $L\to \infty$.  

Still, one may note that for $d=4$ the simple diagrammatic  upper bound on $R_L$ at least does diverge in the limit $L\to \infty$.  At fist sight this may not appear as  adding much information since from \eqref{U4_RCR}  one learns  that in any dimension  
\be \label{univ_R_bound}
0\leq R_L(\beta) \leq 2
\ee 
(which is similar to Glimm and Jaffe's universal  upper bound of the renormalized coupling constant~\cite{GliJaf73}).  
Nevertheless, upon closer inspection, the non-divergence of the upper bound actually provides a meaningful hint.  
To explain the point, let us look closer at what it represents.\\ 
 
 In terms of the random current representation, the all important quotient is bounded by the intersection probability, 
 \be   \label{eq:Rineq} 
\frac{|U_4(x_1,...,x_4)|  }{ \langle\sigma_{x_1} \sigma_{x_2} \rangle \langle \sigma_{x_3} \sigma_{x_4} \rangle}
    \  \leq   \   2\, \,  {\bf P}^{\{x_1, x_2\}, \{x_3, x_4\},\emptyset,\emptyset} 
  [ Q \ne \emptyset]  
 \ee  
where, to abbreviate, we denote $Q:=  {\mathbf C}_{\n_1+\n_3}(x_1)\cap{\mathbf C}_{\n_2+\n_4}(x_3)$.  \\ 

In the diagrammatic bound \eqref{eq:intera}  the probability that $Q $ is non-empty  is  estimated by its expected size, i.e., 
\be 
{\bf P}^{\{x_1, x_2\}, \{x_3, x_4\},\emptyset,\emptyset} [Q\neq \emptyset]  \ \leq \ 
{\bf E}^{\{x_1, x_2\}, \{x_3, x_4\},\emptyset,\emptyset} [ |Q| ]   
\ee 
with $| Q| =  \sum_y  \bm{1}[ y \in Q ] $.  However the exact relation between the two quantities  (given that $|Q|$ is an integer) is
\be  \label{ratio}
{\bf P}  [Q\neq \emptyset]  \ =  \ 
\frac{
{\bf E} \big[ |Q| \big] }
{{\bf E} \big[ |Q| \,\big|  |Q | \geq 1 \big] }\,.
\ee 
In the present situation the numerator is easy to express {\bf exactly} in terms of the two point function, and it is that upper bound to which the scaling estimate  \eqref{D>4} applies.  

Hence, at the marginal dimension gaussianity of the scaling limit would follow if
\be \label{eq:4Dlim}
\lim_{L\to \infty}  {\bf E}_L \big[ |Q| \,\big|  |Q | \geq 1 \big]  = \infty 
\ee 
where the subscript  on ${\bf E}_L $  indicates that the limit refers to configurations $ x_1, x_2, x_3,x_4 $ in which all  distances are of sizes comparable to $L$.  

\subsection{A suggestive dichotomy} \label{sec_dichotomy} \mbox{  } \\ 

Continuing from the last point: at the intuitive level one could expect that if $\eta=0$ (with no relevant logarithmic corrections) then the conditional expectation in \eqref{eq:4Dlim} may diverge due to the fact that $|Q|$ is the sum of contributions  from different scales which should be approximately independent and each  bounded below.   (If, to the contrary $\eta>0$, then  the bound \eqref{D>4} suffices to cover also the  case $d=4$.) 

A bit more explicitly:  the switching lemma suggests that the  expected value of the size of the intersection within $[-\ell,\ell]^2$
of  a pair of RCR cluster which share one of their two  sources, and have the other at distances of order $L\gg \ell$ from the common one  is of the order of 
\be
B_\ell  := \sum_{\substack{ x\in \Z^4 \\  \|x\|< \ell}}  \langle \sigma_0 \sigma_x \rangle_{\beta_c}^2 
\ee 
to which we refer as  the \emph{bubble diagram} on scale $\ell$: 

This in turn suggests that if at the critical point (for simplicity considered first)  the value of $S_2(0,x) \|x\|^2$ shows consistent behavior across scales, i.e. in the scaling regime of interest the following limit exists 
\be 
\kappa =: \lim_{L\to \infty} S_2(0,x) \|x\|^2  
\ee 
then the ratio $U_4/S_4$ in \eqref{D>4} vanishes for one of two complementary reasons: 
\begin{itemize} 
\item  If $\kappa =0 $ then the original tree diagram bound on $U_4/S_4$ for roughly equidistant points tends to zero \\ 
\item   If $\kappa   >0 $ then the bubble diagram diverges, \eqref{eq:4Dlim} holds, and the scaling limit is gaussian due to the divergence of the correcting  factor in the denominator of  \eqref{ratio}.  \\ 
\end{itemize}

\subsection{Logarithmically improved bounds} \mbox{}  \\ 

There are some gaps in the argument outlined above.  Among those:   
\begin{enumerate} 
\item[1)]  what if the behavior of $S_2(0,x) \|x\|^2$ shows  variation in scales, being small on many scales and then non-negligible on other, higher scales, thus avoiding each of the above arguments\\ 
\item[2)]   establish  the suggested relation of $B_L$ with ${\bf E}_L \big[ |Q| \,\big|  |Q | \geq 1 \big] $  for $Q$ the intersection size defined below \eqref{eq:Rineq} \\  
\item[3)]  formulate the analysis so it applies to limits in which also $\beta$ is allowed to vary, as long as $L/\xi(\beta) \to 0$.\\ 
\end{enumerate} 

The first question  was handled by proving a relative monotonicity  (i.e., up to a correction by a constant factor) in scales of $S_2(0,x) \|x\|^{d-2}$ (Theorem 5.6 in \cite{AizDum21}) on which a remark was made below \eqref{crit_decay}.     The other points, for which more substantial analysis was required,  were handled in \cite{AizDum21} through estimates proving abundance of scale on which the two point function satisfies a regularity condition, analysis of the conditional expectation, and proof of sufficient  independence in the contribution of different regular scales to 
${\bf E}_L \big[ |Q| \,\big|  |Q | \geq 1 \big] $.  \\ 

Through such means the following improvement of the tree diagram bound \eqref{eq:interb} was established.   

\begin{theorem}[Improved tree diagram bound ~\cite{AizDum21}]\label{thm:improved tree bound simple} For the n.n.f.~Ising model in dimension $ d=4$, there exist $c,C>0$ such that for every $\beta\le \beta_c$, every $L\le \xi(\beta)$ and every $x,y,z,t\in\Z^d$ at a distance larger than $L$ of each other,
\be\label{eq:improved tree bound}
|U_4^\beta(x,y,z,t)|\le \frac C{B_L(\beta)^c}\sum_{u\in\Z^4} \langle \sigma_u\sigma_x\rangle_\beta  
 \langle \sigma_u\sigma_y\rangle_\beta  
  \langle \sigma_u\sigma_z\rangle_\beta     
   \langle \sigma_u\sigma_t\rangle_\beta,
\ee
\end{theorem}

From this, arguing through the dichotomy described in section \ref{sec_dichotomy} one gets the following improved version of 
\eqref{moment_Ising} 

%
\begin{proposition}\label{prop:gaussian b}
There exist $c,C>0$ such that for the n.n. ferromagnetic~Ising model on $\Z^4$, every $\beta\le\beta_c$, every $L\le\xi(\beta)$, and  test function $f \in C_0(\R^4)$,  
\be \label{improved_MGF_Ising}
\boxed{ \Big|  \langle \exp\big\{  z T_{f,L} \big \}   \rangle - \exp  \big\{\frac{z^2}{2} 
\langle T_{f,L}^2 \rangle  \big \}   \Big |\  
\leq  \      2^{-4} \ z^4   \exp  \big\{\frac{z^2}{2} \langle T_{|f|,L}^2 \rangle  \big \}      \  \frac{\widetilde R_{|f|,L}  \  }{(\log L)^c} }
\ee
with $\widetilde  R_{f,L}(\beta)  \leq r^d \|f\|_\infty^4 R_{rL}(\beta)$
\end{proposition}
Let us recall that the variance of  $T_{f,L}$ is bounded above uniformly in $L$ 
\be 
\langle T_{f,L}^2 \rangle = \frac{\sum_{x,y\in \Lambda(rL)}  \langle \sigma_x   \sigma_y \rangle   \, f(x/L) \,  f(y/L) }
 {\sum_{x,y\in \Lambda(L)}   \langle \sigma_x   \sigma_y \rangle  } 
 \ \leq \   r^d  \|f\|_\infty^2  \, .    
\ee 
If the corresponding scaling limit exists then
\be 
\lim_{L\to \infty}  \langle T_{f,L}^2\rangle = \frac{ \int_{\R^4\times \R^4} f(x) f(y) S_2(x,y) dx \, dy }{ \int_{[-1,1]^4\times [-1,1]^4}  
 S_2(x,y) dx \, dy }
\ee 

The claimed gaussianity of any such scaling limit  follows, since  for any test function $f\in C_0(\R^4)$ the right hand side of \eqref{improved_MGF_Ising} tends to zero as $L\to \infty$, while the variance of $T_{f,L}$ converges to a finite limit,  and for all $z\in \mathbb C$ 
\be
\lim_{L\to \infty}  \langle \exp\big\{  z T_{f,L} \big \}   \rangle = 
 \exp  \big\{\frac{z^2}{2} \lim_{L\to \infty} \langle T_{f,L}^2 \rangle  \big\}  \,. 
\ee   

Through the adjustments discussed in section \ref{sec_extension to GS} the above bounds, and the conclusion, extend also to limits of ferromagnetic lattice $\varphi^4$ systems.\\ 

\bigskip 

\section{Some open questions}  

\noindent {\it 1) Scaling limits in $d=3$  dimensions}  

In closing the above summary let us note the strange coincidence that at present $3D$, which is of obvious relevance for statistical mechanics, is the hardest dimension to reach by rigorous methods.  
The effectiveness of the random current representation in shedding light on the model's behavior   in both $d\geq 4$ and $d=2$ dimensions,  leads one to ask whether the stochastic geometric perspective could be of help also in furthering the understanding of the model's non-trivial behavior in  $3D$. 

Among the approaches to $3D$ which have been explored in physics literature is the expansion in $\eps = 4-d$, starting from the upper critical dimension $d=4$~\cite{WilKad72}.   From the opposite direction, more recently the $3D$ Ising model has also been explored through conformal bootstrap~\cite{ElSk_etal}, which is an extension of an approach that is particularly effective in $2D$.

The constructive field theory program has yielded a non-trivial $\varphi^4$ theory over $\R^3$ (c.f.\cite{BryFroSpe82} and references therein).   The construction shows that, focusing on the ultraviolet regime, for  the weakly coupled $\varphi^4_3$  fields the state of zero renormalized-interaction is unstable.  It would be of interest to supplement this by two sets of results: i) boost the results beyond the perturbative regime to the field's critical manifold,
ii) relate the scaling limits of the critical $\varphi^4$ theories to that of the Ising model. 

This challenge beckons for a merger of the non-perturbative methods deployed in the analysis of the Ising / $\varphi^4$ renormalized coupling constants with some guidance from the perturbative studies of soft  $\lambda \varphi^4_3$ fields.   \\[2ex] 

\noindent{\it 2) Logarithmic corrections} \label{sec_PDI}

As we saw, the existence of logarithmic-type  corrections in $4D$ plays a role in the gaussianity of scaling limit.   It will be of interest to obtain sharper control of such corrections, the like of which are also expected to show up in 
the power law scaling relations in $4D$.  \\

The critical exponent bounds listed in Theorem~\ref{thm:sharpness}  were derived through partial differential inequalities (PDI) on the order parameter $M(\beta,\hat h)$ and its derivative $\chi(\beta,\hat h) = \frac{\partial M(\beta,\hat h)}{\partial \hat h} $ (with $\hat h := \beta h$).  Leading examples are: 
\begin{align} 
 \frac{d}{d \beta}  \chi(\beta) &\, \leq \, \|J\|  \chi(\beta)^2    \notag \\ 
\beta \frac{\partial M}{\partial \beta} & \, \leq\  \beta \| J\| \, \, M \cdot \frac{\partial M}{\partial \widehat h} 
\end{align} 
which hold for both percolation and Ising models and 
\be   \label{M_diag}
M   \leq \widehat h \frac{\partial M}{\partial \widehat h}  \ +\     M^{1+\alpha}    \beta \pder{\beta} M  +  M^{2+\alpha} 
\quad \rm{at} \,\, \alpha=\begin{cases}  1 & \rm{Ising} \\ 
   0 & \rm{percolation}   
\end{cases} 
\ee 
\mbox{ } 

It  was also shown that under certain conditions mildly modified  relations  hold also  in the other direction.    
The conditions allowing such reversal are  the convergence of the bubble diagram $\mathcal{B}(\beta_c)< \infty$ for Ising type models, and the triangle diagram   $\mathcal{T}(\beta_c) < \infty$ for percolation.  These are defined by: 
\begin{align} 
\mathcal{B}  &= (\beta \|J|)^2\sum_x  S_2(0,x) S_2(x,0) = (\beta \|J|)^2 \int_{p\in [-\pi,\pi]^d} \widehat S_2(p)^2 \, dp \notag \\[-2ex]  \\ \notag 
\mathcal{T} &= 
\sum_{x,y}  \tau(0,x) \tau(x,y)  \tau(y,0)  = \int_{p\in [-\pi,\pi]^d} \widehat \tau(p)^3 \, dp\,
\end{align}  
 where $\tau(x,y)$ is the probability of connection in a percolation model and $\widehat{F} $ is the Fourier transform of  $F$. 
 
 By such means it was proven that some of these models' critical exponents assume simple (mean field) values in dimensions $d > d_c$ with 
\be
\boxed{d_c = \begin{cases}  4 & \rm{ Ising} \\ 
 6 & \rm {percolation}\,
 \end{cases}}  \,
\ee 

Can the relation expressed by PDI be sharpened enough to yield the logarithmic corrections at the marginal dimensions?  \\  (If experience is any guide, it could  be of help to consider the two models in parallel.) \\ 

Progress on this question may be closely related to the goal of non-perturbative bounds on the \emph{beta function}  which in the renormalization group analysis guides the flow of the scale-dependent renormalized coupling constant. It would be of interest to link in this manner the current rigorous analysis of the scaling limits with the ``renormalization group'' perspective on the subject.   \\[2ex]

\noindent {\it 3) Can  $P(\varphi)_d$ field theory make sense beyond the realm of convergent functional integrals?} 

It is an intriguing observation  that in the diagrammatic layout of relations such as \eqref{eq:interb} and \eqref{M_diag}    
one typically finds cubic diagrams  for percolation in place of quartic diagrams for Ising models.  
This suggests a relation with
a $ \varphi^3$ versus $ \varphi^4$  field theory, and through that lead one to guess that   
that the upper critical dimension for percolation would be $6$.  
While the euclidean functional integral with  the cubic polynomial in place of \eqref{eq:P4} does not converge in any dimension, 
various aspects of this prediction were proven correct (cf. \cite{HeyHof17} and references therein).

One wonders: could there be another perspective on field theory, in which it makes sense beyond the somewhat timid formulation which is outlined in the introduction?  Indeed different proposals, motivated by other observations, continue being made   cf.~\cite{FanKla21}.  

\bigskip

\noindent {\bf Acknowledgments}  I would like to thank Hugo Duminil-Copin for our  fruitful collaboration.  It was advanced through mutual visits to Princeton and Geneva University,  sponsored by a  Princeton-UniGe  partnership grant.  I am also grateful for the earlier successful collaborations on the subject with R. Graham, D. Barsky, R. Fernandez and S. Warzel, and thank J. Shapiro and S. Warzel for useful comments on the initial draft of this article.
 The author's  work reported here was supported in parts by  NSF grants, and the Weston visiting professorship at the Weizmann Institute.

\providecommand{\bysame}{\leavevmode\hbox to3em{\hrulefill}\thinspace}
\providecommand{\MR}{\relax\ifhmode\unskip\space\fi MR }
\providecommand{\MRhref}[2]{
  \href{http://www.ams.org/mathscinet-getitem?mr=#1}{#2}
}
\providecommand{\href}[2]{#2}

\end{document}